\documentclass[11pt]{article}
\usepackage[margin=1.5in]{geometry}

\usepackage{graphicx}
\usepackage{subfig}
\usepackage{hyperref,doi}
\usepackage{cite}
\usepackage[pagewise]{lineno}
\usepackage{amsmath,amssymb}

\newtheorem{lemma}{Lemma}
\newtheorem{theorem}{Theorem}

\newtheorem{corollary}{Corollary}
\newtheorem{invariant}{Invariant}
\newcommand{\qed}{\hspace*{\fill}\rule{6pt}{6pt}}
\newenvironment{proof}{\noindent{\bf Proof:}}{\bigskip} 
\makeatletter
\def\@begintheorem#1#2{\sl \trivlist \item[\hskip \labelsep{\bf #1\ #2:}]}
\def\@opargbegintheorem#1#2#3{\sl \trivlist
      \item[\hskip \labelsep{\bf #1\ #2\ #3:}]}
\makeatother

\newcommand{\highlight}[1]{{\itshape #1}}

\begin{document}

	\title{Drawing Trees with Perfect Angular Resolution\\ and
	  Polynomial Area\footnote{A preliminary version of this paper appeared at the 18th International Symposium on Graph Drawing (GD'10)~\cite{degkn-dtwpa-10}.}}

	\author{Christian A. Duncan\thanks{Department of Mathematics and Computer Science, Quinnipiac University, Hamden, CT, USA} \and
	  David Eppstein\thanks{Department of Computer Science, University of California, Irvine, California, USA} \and
	  Michael T. Goodrich$^\ddagger$ \and
	  Stephen G. Kobourov\thanks{Department of Computer Science, University of Arizona, Tucson, Arizona, USA} \and
	  Martin N\"ollenburg\thanks{Institute of Theoretical Informatics, Karlsruhe Institute of Technology, Germany}}

\date{}
\maketitle

\begin{abstract}
	We study methods for drawing
	trees with perfect angular resolution, i.e.,
	with angles at each node $v$ equal to $2\pi/d(v)$.
	We show:
	\begin{enumerate}
	\item
	Any unordered tree has a crossing-free straight-line drawing
	with perfect angular resolution and polynomial area.
	\item
	There are ordered trees that require exponential area for any
	crossing-free straight-line drawing having perfect angular
	resolution.
	\item
	Any ordered tree has a crossing-free \highlight{Lombardi-style} drawing 
	(where each edge is represented by a circular arc)
	with perfect angular resolution and polynomial area.
	\end{enumerate}
	Thus, our results explore what is achievable with straight-line drawings and 
	what more is achievable with
	Lombardi-style drawings, with respect to drawings of trees with
	perfect angular resolution.
\end{abstract}

\newpage

\section{Introduction}
\label{sec:intro}

Most methods for visualizing trees aim to produce drawings
that meet as many of the following aesthetic constraints as possible:
\begin{list}{\labelenumi}{
\usecounter{enumi}
\setlength{\leftmargin}{5em}
\setlength{\labelwidth}{3em}
\setlength{\itemindent}{0em}}
\item straight-line edges,
\item crossing-free edges,
\item polynomial area, and
\item perfect angular resolution around each node.
\end{list}
These constraints are all well-motivated, in that we 
desire edges that are easy to follow, do not confuse viewers with
edge crossings,
are drawable using limited real estate, and avoid congested 
incidences at nodes.
Nevertheless, previous tree drawing algorithms have made various compromises
with respect to this set of constraints; we are not aware of any previous tree-drawing algorithm
that can achieve all these goals simultaneously.
Our goal in this paper is to show what is actually possible with respect to
this set of constraints and to expand it 
further with a richer notion of edges that are easy to follow.
In particular, we desire tree-drawing algorithms that satisfy
all of these constraints simultaneously.
If this is provably not possible,
we desire an augmentation that avoids compromise and instead 
meets the spirit of all of these goals in a new way, which, in the
case of this paper, is
inspired by the work of artist Mark Lombardi~\cite{hl-mlgn-03}.

\paragraph{Problem Statement.}
The art of Mark Lombardi involves drawings of social networks,
typically using circular arcs and good angular resolution.
Figure~\ref{fig:lombardi} shows such a work of Lombardi that is
crossing-free and almost a tree.
It makes use of both circular arcs and straight-line
edges.
Inspired by this work, let us define a set of problems that explore
what is achievable for drawings of trees with respect to the
constraints listed above but that, like Lombardi's drawings, also allow curved as well as straight-line edges.

\begin{figure}[hbt!]
\centering
\includegraphics[width=3.9in]{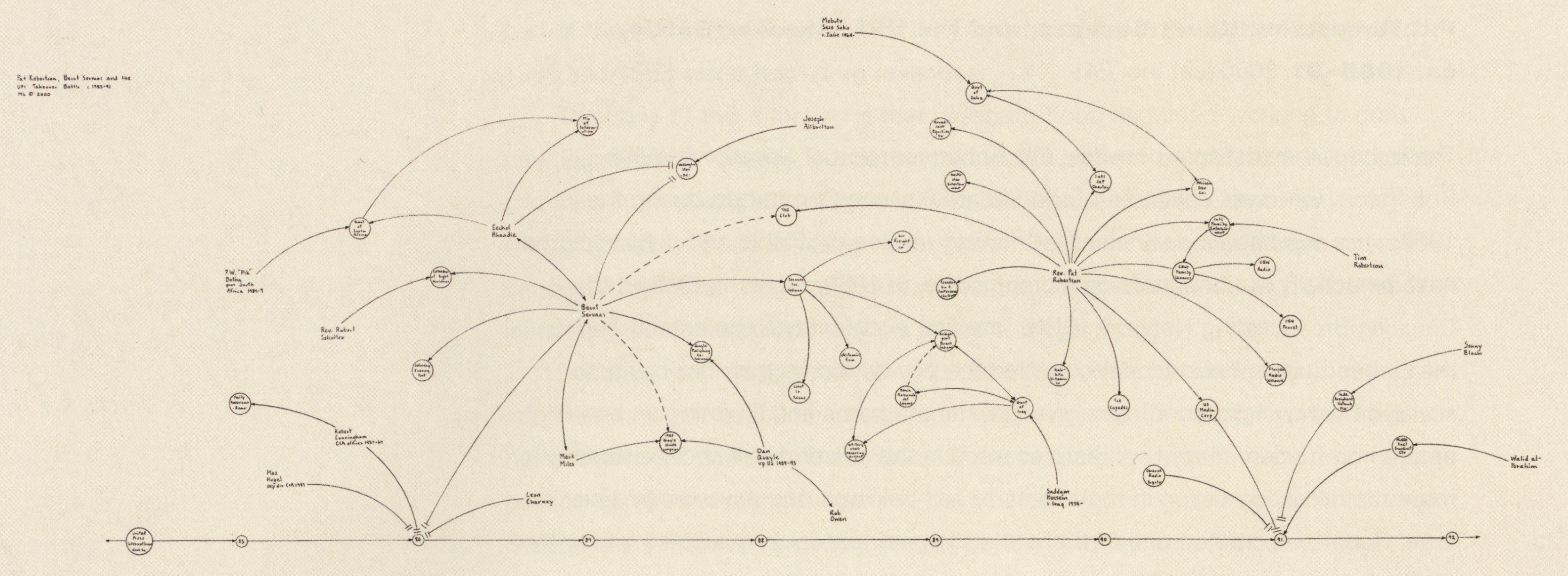}
\caption{\label{fig:lombardi} \textit{Pat Robertson,
Beurt Servaas and the UPI Takeover Battle, ca. 1985-91.} Drawing by Mark Lombardi, 2000. Image courtesy of Pierogi.}
\end{figure}

A \highlight{drawing} of a graph $G=(V,E)$ 
is an assignment of a unique point in the Euclidean plane to each node in $V$
and an assignment of a simple curve to each edge $(u,v) \in E$
such that the only two nodes in $V$ intersected by the curve are $u$ and $v$, which coincide with the endpoints of the curve.
A drawing is \highlight{straight-line} if 
every edge is drawn as a straight-line segment.
A drawing is \highlight{planar} if no two curves intersect 
except at a common shared endpoint.

Given a graph $G=(V,E)$,
let $d(u)$ denote the
\highlight{degree} 
of a node $u$, i.e., the number of edges incident to $u$ in $G$.
For a drawing of $G$, the \highlight{angular resolution} at a 
node $u$ is the minimum angle between any two edges incident to $u$.
A node has \highlight{perfect angular resolution} 
if its angular resolution is ${2\pi}/{d(u)}$, and
a drawing has perfect angular resolution if {\em every} node does.

Suppose that our input graph $G$ 
is a rooted tree $T$.
We say that $T$ is \highlight{ordered} if an ordering of the edges
incident to each node in $T$ is specified. Otherwise, $T$ is 
\highlight{unordered}.

In many drawings of graphs, nodes can be placed on an integer grid,
allowing one to get a bound on the area of the drawing by bounding
the dimensions of the grid.
Drawings with perfect angular resolution
cannot be placed on an integer 
grid unless the degrees of the nodes are constrained.
To see this, suppose we have a vertex $u$ and 
two of its (consecutive) neighbors all of 
which lie on Cartesian grid points.
From basic trigonometry, the area of the triangle defined by 
these points is $\frac{1}{2}ab\sin{\theta}$,
where $a$ and $b$ represent the lengths of the edges 
extending from $u$ and $\theta=2\pi/d(u)$ is the angle between 
these two edges.
By Pick's theorem, the area of this triangle is rational,
and consequently so is the square of the area.
Since $a^2$ and $b^2$ must also be rational, we conclude that
$\sin^2{\theta}$ must be rational.
This is false for nearly all values of $d(u)$,   
for example, when $d(u)=10$ and $\theta=\pi/5$.
Hence, if we wish to have perfect angular resolution,
we cannot require the nodes to have integer coordinates.

In this paper, our focus is on producing planar drawings of trees with perfect angular resolution in \emph{polynomial area}.
When defining the area of a drawing, it is important that the area measure prevents the drawing from being arbitrarily scaled down.
Our algorithms achieve polynomial area bounds according to the following three typical area measures for non-grid drawings.
In the first measure, the area is defined as the ratio of the area of a smallest disk enclosing the drawing to the square of the length of its shortest edge.
As two non-neighboring nodes can be arbitrarily close using this definition, one may be interested in using another definition of area instead, the (squared) ratio of the farthest pair of nodes to the closest pair of nodes in the drawing.
This area measure can also be defined in terms of \emph{edges} instead of \emph{nodes}, i.e., as the (squared) ratio of the farthest pair of edges to the closest pair of non-adjacent edges.

We define a \highlight{Lombardi drawing}~\cite{degkn-ldg-12} of a graph $G$ as a drawing of $G$
with perfect angular resolution such that each edge is drawn as a circular arc.
When measuring the angle formed by two circular arcs incident to a node $v$, 
we use the angle formed by the tangents of the two arcs at $v$.
Circular arcs are strictly more general than straight-line segments, since straight-line segments can be viewed as circular arcs
with infinite radius. Figure~\ref{fig:examples} shows an example of a
straight-line drawing and a Lombardi drawing for the same tree.
Thus, we can define our problems as follows:
\begin{enumerate}
\setlength{\itemsep}{2pt}
\item
Is it always possible to produce a straight-line drawing
of an unordered tree with perfect 
angular resolution and polynomial area?
\item
Is it always possible to produce a straight-line drawing
of an ordered tree with perfect 
angular resolution and polynomial area?
\item
Is it always possible to produce a Lombardi drawing
of an ordered tree with perfect 
angular resolution and polynomial area?
\end{enumerate}

\begin{figure}[tb]
  \centering
  \subfloat[Straight-line drawing for an unordered tree\label{fig:exmpl-straight}]{\includegraphics[page=1,scale=.48]{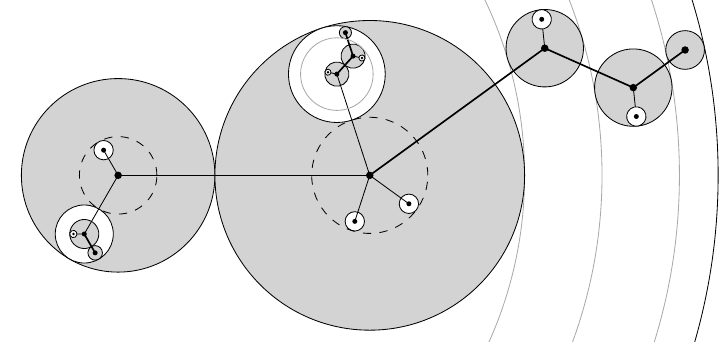}}
  \hfill
  \subfloat[Lombardi drawing for an ordered tree\label{fig:exmpl-lombardi}]{\includegraphics[page=3,scale=.48]{example-straight}}
  \caption{Two drawings of a tree $T$ with perfect
    angular resolution and polynomial area as produced by our
    algorithms. Bold edges are heavy edges, gray disks are heavy
    nodes, and white disks are light children. The root of $T$ is in
    the center of the leftmost disk.}
  \label{fig:examples}
\end{figure}

\paragraph{Related Work.}
Tree drawings have interested researchers 
for many decades: e.g., hierarchical drawings of binary trees date to the 1970's~\cite{ws-td-79}. 
Many improvements have been proposed since this early work, using space efficiently and generalizing to 
non-binary trees\cite{bj-iw-02,cgkt-oaars-02,ggt-putdo-96,%
DBLP:journals/ijcga/GargR03,%
DBLP:journals/jgaa/GargR04,rt-td-81,%
Shin2000175,w-np-90}.  
These drawings fail to meet the four constraints mentioned earlier, especially the constraint on angular resolution.

Several other methods directly aim to optimize angular
resolution in tree drawings.
Radial drawings of trees place nodes at the same distance 
from the root on a circle 
around the root node~\cite{e-df-92}.
Circular tree drawings are made of recursive 
radial-type layouts~\cite{mh-cdrt-98}. 
Bubble drawings~\cite{ga-bt-04} draw trees recursively with each subtree
contained within a circle disjoint from its siblings but
within the circle of its parent.
Balloon drawings~\cite{ly-bdrt-07} take a similar approach
and heuristically attempt to optimize space 
utilization and the ratio between the longest and shortest edges in the tree.
Convex drawings~\cite{ce-tcfoa-07} partition the plane into unbounded convex polygons with their boundaries formed by tree edges.
Although these methods provide several benefits, 
none of these methods guarantees that they satisfy all of 
the aforementioned constraints.

The notion of drawing graphs with edges that are circular arcs
or other nonlinear curves is certainly not new to graph
drawing. For instance, Cheng~{\em et al.}~\cite{cdgk-dpgca} use circular arcs
to draw planar graphs in an $O(n) \times O(n)$ grid
while maintaining bounded (but not perfect) angular resolution.
Similarly, Dickerson~{\em et al.}~\cite{degm-cdvnp-03} use circular-arc
polylines to produce planar confluent drawings of non-planar graphs,
Duncan {\em et al.}~\cite{DBLP:journals/ijfcs/DuncanEKW06} draw
graphs with fat edges that include circular arcs, and
Cappos {\em et al.}~\cite{cefk-sgeb-09} study simultaneous embeddings
of planar graphs using circular arcs.
Finkel and Tamassia~\cite{ft-cgduf-04} use a
force-directed method for producing curvilinear drawings, and Brandes and Wagner~\cite{bw-uglvt-00} use energy minimization methods to place B\'ezier splines that represent connections in a train network. 

In a separate paper~\cite{degkn-ldg-12} we study Lombardi drawings for classes of graphs other than trees. Unlike trees, not all planar graphs have planar Lombardi drawings~\cite{degkn-ldg-12,degkl-ppld-12} and it is an interesting open question to characterize the graphs that have a planar Lombardi drawing. Eppstein~\cite{e-pldsg-12} recently proved that all planar subcubic graphs have a planar Lombardi drawing, and that there are 4-regular planar graphs that do not have a planar Lombardi drawing. He also characterized the planar graphs that have planar Lombardi drawings corresponding to physical soap bubble clusters~\cite{e-gpsb-12}. L\"offler and N\"ollenburg~\cite{ln-pldo-12} showed that all outerpaths, i.e., outerplanar graphs whose weak dual is a path, have an outerplanar Lombardi drawing. In terms of the usability of Lombardi drawings, two independent user studies~\cite{phnk-ulgd-12,xrph-uscegv-12} examined the performance of Lombardi versus straight-line drawings for several graph reading tasks. While the study of Purchase \emph{et al.}~\cite{phnk-ulgd-12} showed an advantage  of straight-line drawings for two out of three tasks, but aesthetic preference for Lombardi drawings, the study of Xu \emph{et al.}~\cite{xrph-uscegv-12} did not show significant performance differences between the two types of drawings, but a strong aesthetic preference for straight-line drawings.

\paragraph{Our Contributions.} 
In this paper we present the first algorithm for producing
straight-line, crossing-free 
drawings of unordered trees that ensures perfect angular resolution and 
polynomial area. 
In addition we show, in Section~\ref{sec:exponentialArea},
that if the tree is ordered then it 
is not always possible to maintain perfect angular 
resolution and polynomial drawing area
when using straight lines for edges.
Nevertheless,
in Section~\ref{sec:lombardi}, we show that crossing-free polynomial-area
Lombardi drawings of ordered trees are possible.
That is, we show that
the answers to the questions posed above are ``yes,'' ``no,''
and ``yes,'' respectively.
Both algorithms require linear time in a model of computation, in which we can perform trigonometric computations and find roots of bounded degree polynomials in constant time.

\section{Straight-line drawings for unordered trees}\label{sec:straight-line-tree}
\label{sec:straightLine}

Let $T$ be an unordered tree with $n$ nodes.  We wish to construct a
straight-line drawing of $T$ with perfect angular resolution and
polynomial area.

The main idea of our algorithm is, similarly to the common bubble and
balloon tree constructions~\cite{ga-bt-04,ly-bdrt-07}, to draw the
children of each node of the given tree in a disk centered at that
node; however, our algorithm differs in several key respects in order
to achieve the desired area bounds and perfect angular resolution.

\subsection{Heavy Path Decomposition}\label{sec:hpd}

The initial step before drawing the tree $T$ is to create a heavy path
decomposition~\cite{ht-fafnc-84} of $T$.  To make the analysis
simpler, we assume $T$ is rooted at some arbitrary node $r$.  We let
$T_u$ represent the subtree of $T$ rooted at $u$, and $|T_u|$ the
number of nodes in $T_u$. A node $c$ is the \highlight{heavy child} of
$u$ if $|T_c| \geq |T_v|$ for all children $v$ of $u$.  In the case of
a tie, we arbitrarily designate one node as the heavy child.  We refer
to the non-heavy children as \highlight{light} and let $L(u)$ denote
the set of all light children of $u$.  The \highlight{light subtrees}
of $u$ are the subtrees of all light children of $u$. We define $l(u)
= 1+\sum_{v \in L(u)} |T_v|$ to be the \highlight{light size} of
$u$. An edge is called a \highlight{heavy edge} if it connects a heavy
child to its parent; otherwise it is a \highlight{light edge}. The set
of all heavy edges creates the \highlight{heavy-path decomposition} of
$T$, a disjoint set of (heavy) paths where every node in $T$ belongs
to exactly one path (possibly of length 0); see Figure~\ref{fig:heavyPath}. 
After an initial bottom-up traversal of $T$ to compute the number of descendants for every node, the heavy-path decomposition can be computed by a depth-first search that always descends to the heavy child of each node before visiting its light children in arbitrary order. 
This takes $O(n)$ time.

The heavy path
decomposition has the following important property.  If we treat
each heavy path as a node, and each light edge as connecting two
heavy-path nodes, we obtain a tree $H(T)$. This tree has height $h(T)
\le \log_2{n}$ since the size of each light child is less than half the
size of its parent. We refer to the \highlight{level} of a heavy path
as the depth of the corresponding node in the decomposition tree,
where the root has depth 0. We extend this notion to nodes, i.e., the
level of a node $v$ is the level of the heavy path to which $v$
belongs.
  \begin{figure}[tb]
    \begin{center}
      \includegraphics[scale=.75,page=2]{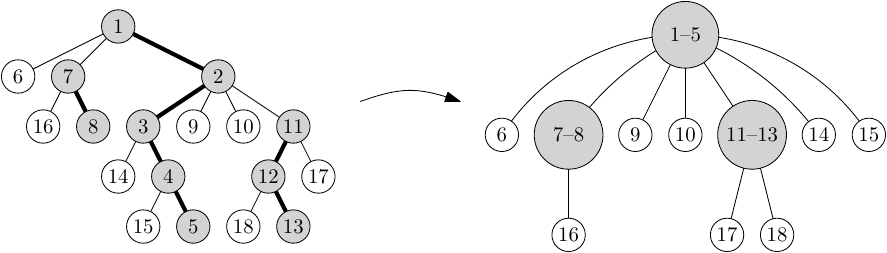}
    \end{center}
    \caption{The tree $T$ on the left highlights its heavy edges. The
      corresponding heavy-path decomposition tree $H(T)$ on the right has each
      heavy path represented by a single node.}
    \label{fig:heavyPath}
  \end{figure}

\subsection{Drawing Algorithm}
\label{sec:drawalg}

Our algorithm draws $T$ incrementally in the order of a depth-first
traversal of the corresponding heavy-path decomposition tree $H(T)$,
i.e., given drawings of the light subtrees of a heavy-path node $P$ in
$H(T)$ we construct a drawing of $P$ and its subtrees. Let $P = (v_1,
\ldots, v_k)$ be a heavy path. Then we draw each node $v_i$ of $P$ in
the center of a disk $D_i$ and place smaller disks containing the
drawings of the light children of $v_i$ and their descendents around
$v_i$ in two concentric annuli of $D_i$. We guarantee perfect angular
resolution at $v_i$ by connecting the centers of the child disks with
appropriately spaced straight-line edges to $v_i$. Next, we create the
drawing of $P$ and its descendents within a disk $D$ by placing $D_1$
in the center of $D$ and $D_2, \dots, D_k$ on concentric circles
around $D_1$. We show that the radius of $D$ is linear in the number
$n(P)$ of nodes descending from $P$ and exponential in the level of
$P$. In this way, at each step downwards in the heavy path
decomposition, the total radius of the disks at that level shrinks by
a constant factor, allowing room for disks at lower levels to be
placed within the higher-level disks. Figure~\ref{fig:exmpl-straight}
shows a drawing of an unordered tree according to our method.

Before we can describe the details of our construction we need the
following geometric property.  Define an
\highlight{$(R,\delta)$-wedge}, $\delta \le \pi$ as a sector of angle
$\delta$ of a radius-$R$ disk; see Figure~\ref{fig:circleFit}.

\begin{lemma}
\label{lemma:circleFit}
The largest disk that fits inside an $(R,\delta)$-wedge has
radius $r =
R\frac{\sin(\delta/2)}{1+\sin(\delta/2)}$.
\end{lemma}

\begin{proof}
  The largest disk inside the $(R,\delta)$-wedge
  touches the circular arc and both radii of the wedge. Thus we
  immediately obtain a right triangle formed by the apex of the wedge,
  the center of the disk we want to fit, and one of its tangency
  points with the two radii of the wedge; see
  Figure~\ref{fig:circleFit}. This triangle has one side of length $r$
  and hypothenuse of length $R-r$. From $\sin (\delta/2) = \frac{r}{R-r}$ we
  obtain $r = R\frac{\sin(\delta/2)}{1+\sin(\delta/2)}$. \qed
  \begin{figure}[htb]
    \centering
    \includegraphics[]{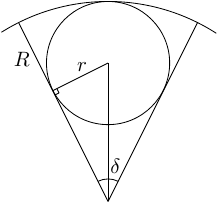}
    \caption{An
      $(R,\delta)$-wedge and the largest disk that can be placed inside it.}
    \label{fig:circleFit}
  \end{figure}
\end{proof}

In the next lemma we show how to draw a single node $v$ of a heavy
path $P$ given drawings of all its light subtrees.
\begin{lemma}\label{lem:heavynode:straight}
  Let $v$ be a node of $T$ at level $j$ of $H(T)$. For each light child $u \in L(v)$
  assume there is a disk $D_u$ of radius $r_u = 2\cdot 8^{h(T)-j-1}
  |T_u|$ that contains a fixed drawing of $T_u$ with perfect angular
  resolution and such that $u$ is in the center of~$D_u$. Then we can
  construct a drawing of $v$ and its light subtrees inside a disk $D$ in~$O(d(v))$ time such that the following properties hold:
  \begin{enumerate}
  \item\label{dsle:item:edge} the edge between $v$ and any light child $u
    \in L(v)$ is a straight-line segment that does not intersect any
    disk other than $D_u$;
  \item\label{dsle:item:heavyedge} one or two rays that do not intersect any disk $D_u$ are reserved for drawing the heavy edges incident to $v$ or the light edge to the parent of $v$;
  \item\label{dsle:item:disks} any two disks $D_u$ and $D_{u'}$ for two light children $u \ne u'$ are disjoint;
  \item\label{dsle:item:angres} the angular resolution of $v$ is $2\pi/d(v)$;
  \item\label{dsle:item:angle} the angle between the two rays reserved for the heavy edges or the light parent edge is at least
    $2\pi/3$ and at most $4\pi/3$ (if these two rays exist);
  \item\label{dsle:item:area} the disk $D$ has radius $r_v = 8^{h(T)-j} l(v)$.
  \end{enumerate}
\end{lemma}
\begin{proof}
  We assume that the ray $\rho_0$ for the (heavy or light) edge to the parent of $v$ is directed
  horizontally to the left (for the root of $T$ its unique heavy edge takes this role). We draw a disk $D$ with radius $r_v$
  centered at $v$ and create $d(v)$ \highlight{spokes}, i.e., rays
  extending from $v$, that are equally
  spaced by an angle of $2\pi/d(v)$ and include the ray $\rho_0$. In order to achieve the angular resolution (property~4), every neighbor of $v$
  must be placed on a distinct spoke. The main difficulty is that there can be child disks
  that are too large to place without overlap on adjacent spokes
  inside $D$.

  Let $D_{\max}$ be the largest disk $D_u$ of any $u \in L(v)$ and let
  $r_{\max}$ be its radius. We split $D$ into an outer annulus $A$ and
  an inner disk $B$ by a concentric circle of radius $R = r_v -
  2r_{\max}$; see Figure~\ref{fig:light-straight}. We define a child
  $u \in L(v)$ to be a \highlight{small} child, if  $r_u \le R
  \frac{\sin (\pi/d(v))}{1+ \sin (\pi/d(v))}$, and to be a
  \highlight{large} child otherwise. We further say $D_u$ is a small
  (large) disk if $u$ is a small (large) child.  We denote the number
  of small children as $n_s$ and the number of large children as
  $n_l$.  By Lemma~\ref{lemma:circleFit} we know that any small disk
  $D_u$ can be placed inside an $(R,2\pi/d(v))$-wedge. This means that
  we can place all $n_s$ small disks centered on any subset of $n_s$ spokes
  inside $B$ without violating property~\ref{dsle:item:disks}. So once we
  have placed all large disks correctly then we can always distribute
  the small children on the unused spokes.

  \begin{figure}[tb]
    \centering
    \subfloat[All light subtrees fit into a disk of radius $r_v/4$
    and are split into small and large disks.\label{fig:small-large-children}]{\makebox[.48\textwidth]{\includegraphics[page=1,scale=.33]{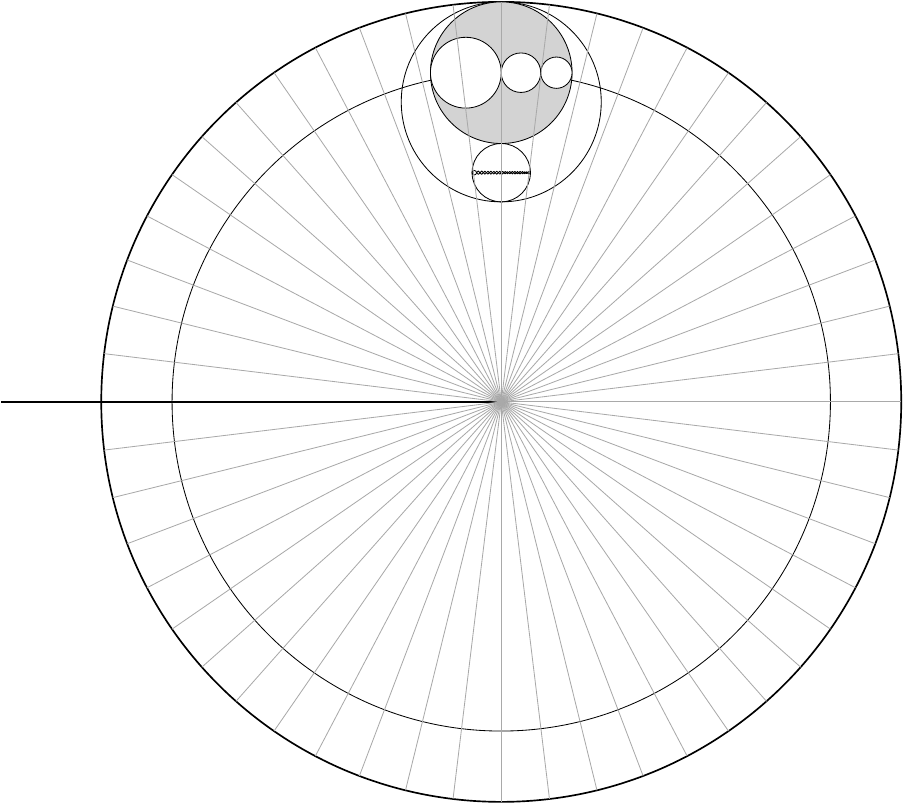}}} 
    \hfill
    \subfloat[Large disks are placed in the outer annulus and small disks in
    the inner disk.]{\includegraphics[page=2,scale=.33]{straight-tree}}
    \caption{Drawing a node $v$ and the subtrees of its light children $L(v)$.}
    \label{fig:light-straight}
  \end{figure}

  We place all large disks in the outer annulus $A$. Observe
  that \[4 \sum_{u \in L(v)} r_u = 4 \sum_{u \in L(v)} 2\cdot
  8^{h(T)-j-1} |T_u| = 8^{h(T)-j} \sum_{u \in L(v)} |T_u| < 8^{h(T)-j}
  l(v) = r_v,\] i.e., we can place all light children on the diameter
  of a disk of radius at most $r_v/4$. If we order all light children
  along that diameter by their size we can split them into one disk
  containing the large disks and one containing the small disks;
  see Figure~\ref{fig:small-large-children}.
  
  Assume that the large disks are arranged on the horizontal diameter
  of their disk and that this disk is placed vertically above $v$ and
  tangent to $D$ as shown in
  Figure~\ref{fig:small-large-children}. Since that disk has radius at
  most $r_v/4$ we can use Lemma~\ref{lemma:circleFit} to show that it
  always fits inside an $(r_v,\pi/4)$-wedge. If we now translate the
  large disks vertically upward onto a circle centered at $v$ with
  radius $r_v - r_{\max}$ then they are still disjoint and they all
  lie in the intersection of $A$ and the $(r_v,\pi/4)$-wedge. We now
  rotate them counterclockwise around $v$ until the leftmost disk
  $D_{\max}$ touches the ray $\rho_0$. Thus all large disks
  are placed disjointly inside a $\pi/4$-sector of $A$. However, they
  are not centered on the spokes yet.

  Beginning from the leftmost large disk, we rotate each large disk
  $D_u$ and all its right neighbors clockwise around $v$ until $D_u$
  snaps to the next available spoke. Clearly, in each of the $n_l$
  steps we rotate by at most $2\pi/d(v)$ in order to reach the next
  spoke.

  We now bound the number $n_l$ of large children. By definition a
  child is large if $r_u = 2\cdot 8^{h(T)-j-1} |T_u| > (r_v-2r_{\max})
  \frac{\sin (\pi/d(v))}{1+ \sin (\pi/d(v))}$. We also have $r_v \ge
  8^{h(T)-j} \sum_{u \in L(v)} |T_u|$. Let $w$ be the light child of
  $v$ with maximum disk radius $r_w = r_{\max}$. Then $r_w = 2 \cdot
  8^{h(T)-j-1} |T_w|$ and hence $r_v - 2r_{\max} \ge 4\cdot
  8^{h(T)-j-1} (2\sum_{u \in L(v)}|T_u| - |T_w|)$. So for a light
  child $u$ to be large, its subtree $T_u$ has to contain $|T_u| > 2
  \cdot (2\sum_{u \in L(v)}|T_u| - |T_w|) \frac{\sin (\pi/d(v))}{1+
    \sin (\pi/d(v))}$ nodes. This yields
  \[
  n_l < 1 + \frac{\sum_{u \in L(v)}|T_u| - |T_w|}{2 \cdot (2\sum_{u
      \in L(v)}|T_u| - |T_w|) \frac{\sin (\pi/d(v))}{1+ \sin
      (\pi/d(v))}} < 1 + \frac{1+ \sin (\pi/d(v))}{4 \sin (\pi/d(v))}.
  \]
  
  From this we obtain that for $d(v) \ge 5$ we have $n_l <
  3d(v)/8$. So for $d(v) \ge 5$ we can always place all large disks
  correctly on spokes inside at most half of the outer annulus $A$
  since we initially place all large disks in a $\pi/4$-wedge and then
  enlarge that wedge by at most $3d(v)/8 \cdot 2\pi/d(v) = 3\pi/4$
  radians. For $d(v) \le 2$ there are no light children, for $d(v)=3$ we
  immediately place the disk of the single light child on its spoke without
  intersecting the other spokes, and for $d(v)=4$ we place the disks of the two light children on opposite vertical spokes separated by the two
  horizontal spokes, which does not produce any intersections either. If $v$ is the root of $T$ and $d(v) \le 4$ the disks of the light children (at most three) are placed analogously.

  Since we require at most half of $A$ to place all large children, we
  can assign the second ray for a heavy edge (if it exists) to the spoke exactly opposite of
  $\rho_0$ if $d(v)$ is even. If $d(v)$ is odd, we choose
  one of the two spokes whose angle with $\rho_0$ is
  closest to $\pi$. 
  Finally, we arbitrarily assign the $n_s$ small children to the
  remaining free spokes inside the inner disk $B$.

  Thus the drawing for $v$ and its light subtrees constructed in this fashion satisfies properties 1--6. 

	It remains to show that the drawing can be constructed in $O(d(v))$ time. 
	In order to avoid unnecessary updates of the node coordinates, we store the position of each node (in polar coordinates) relative to its parent, i.e., relative to $v$. 
	Thus we can change the placement of the whole subtree $T_v$ by changing only the position of its root node $v$.
	We first assign the large children in arbitrary order to their spokes. 
	The next feasible spoke is easily obtained from the position and radius of the previous disk and the radius of the next disk.
	Then we place the small children on the remaining spokes and reserve the stub for the heavy child.
	It is sufficient to assign a unique spoke ID in $\{2, 3, \dots, d(v)\}$ to each child, where spoke~1 connects to the parent of $v$. 
	This spoke order can be interpreted both clockwise and counterclockwise, which will be useful for drawing the heavy paths in the next step.
  Since the placement of any child disk requires constant time, the $O(d(v))$ time bound follows. \qed
\end{proof}

Lemma~\ref{lem:heavynode:straight} shows how to draw a single heavy
node $v$ and its light subtrees. It also applies to the root of $T$ if
we ignore the incoming heavy edge, and to the root node $v_1$ of a
heavy path $P=(v_1, \ldots, v_k)$ at level $l \ge 1$ if we consider
the light edge $uv_1$ to its parent $u$ as a heavy edge for $v_1$. The last node $v_k$ of $P$ is always a leaf, which is trivial
to draw. For drawing an entire heavy path $P = (v_1, \ldots, v_k)$ we
need to link the drawings of the heavy nodes into a path.

\begin{lemma}\label{lem:heavypath:straight}
  Given a heavy path $P=(v_1, \ldots, v_k)$ and a drawing for each
  $v_i$ and its light subtrees inside a disk $D_i$ of radius $r_i$, we
  can draw $P$ and all its descendants inside a disk $D$ in $O(k)$ time such that the
  following properties hold:
  \begin{enumerate}
  \item\label{dsle:item:hp:edges} the heavy edge $v_iv_{i+1}$ is a
    straight-line segment that does not intersect any disk other than
    $D_i$ and $D_{i+1}$;
  \item\label{dsle:item:hp:light} the light edge connecting $v_1$ and its
    parent does not intersect the drawing of $P$;
  \item\label{dsle:item:hp:disjoint} any two disks $D_i$ and $D_j$ for
    $i\ne j$ are disjoint;
  \item\label{dsle:item:hp:angres} the drawing has perfect angular
    resolution;
  \item\label{dsle:item:hp:area} the radius $r$ of $D$ is $r=2\sum_{i=1}^k r_i$.
  \end{enumerate}
\end{lemma}
\begin{proof}
  Let $v_1$ be the root of $P$ and let $u$ be the parent of $v_1$
  (unless $P$ is the heavy path at level 0). We place the disk $D_1$
  at the center of $D$ and assume that the edge $uv_1$ extends
  horizontally to the left. We create $k-1$ vertical strips $S_2,
  \ldots, S_k$ to the right of $D_1$, each $S_i$ of width $2r_i$; see
  Figure~\ref{fig:v-strips}. Each disk $D_i$ will be placed inside its
  strip $S_i$. We extend the ray induced by the stub
  reserved for the heavy edge $v_1v_2$ from $v_1$ until it intersects the
  vertical line bisecting $S_2$ and place $v_2$ at this intersection
  point. By property~\ref{dsle:item:angle} of
  Lemma~\ref{lem:heavynode:straight} we know that the angle between
  the two heavy edges incident to a heavy node is between $2\pi/3$ and
  $4\pi/3$. Thus $v_2$ is inside a right-open $2\pi/3$-wedge $W$ that
  is symmetric to the $x$-axis. Now for $i=2, \ldots, k-1$ we extend
  from $v_i$ the stub of the heavy edge $v_iv_{i+1}$ into a ray and
  place $v_{i+1}$ at the intersection of that ray and the bisector of
  $S_{i+1}$. When placing the disk $D_{i+1}$ centered at $v_{i+1}$, Lemma~\ref{lem:heavynode:straight} leaves the two valid options of arranging the subtrees of $v_{i+1}$ inside $D_{i+1}$ in clockwise or counterclockwise order. 
We pick the ordering for which the slope of the ray $v_{i+1}v_{i+2}$ is closer to 0, i.e., $v_{i+1}v_{i+2}$ makes a right turn if $v_iv_{i+1}$ has a positive slope and a left turn otherwise. 
(If $\angle v_iv_{i+1}v_{i+2} = \pi$ either way is fine.)
Then by using induction and property~\ref{dsle:item:angle} of Lemma~\ref{lem:heavynode:straight} the ray $v_{i+1}v_{i+2}$ stays within~$W$.

  \begin{figure}[tb]
    \centering
    \subfloat[Placing disks in vertical strips.\label{fig:v-strips}]{\includegraphics[page=1,scale=.9]{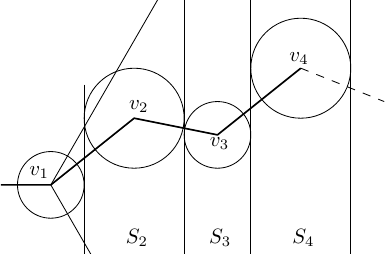}}
    \hfill
    \subfloat[Final transformation of the drawing.\label{fig:strips2annuli}]{\includegraphics[page=2,scale=.9]{v_strips}}
    \caption{Constructing the heavy path drawing by appending drawings of its heavy nodes.}
    \label{fig:straight-path}
  \end{figure}

  Since each disk $D_i$ is placed in its own strip $S_i$, no two disks
  intersect (property~\ref{dsle:item:hp:disjoint}) and since heavy edges
  are straight-line segments within two adjacent strips, they do not
  intersect any non-incident disks (property~\ref{dsle:item:hp:edges}). The
  light edge $uv_1$ is completely to the left of all strips and thus
  does not intersect the drawing of $P$
  (property~\ref{dsle:item:hp:light}). Since we were using the existing
  drawings (or their mirror images) of all heavy nodes, their perfect
  angular resolution is preserved (property~\ref{dsle:item:hp:angres}). 

  The current drawing has a width that is equal to the sum of the
  diameters of the disks $D_1, \ldots, D_k$. However, it does not yet
  necessarily fit into a disk $D$ centered at $v_1$ whose radius
  equals that sum of the diameters. To achieve this we create $k-1$
  annuli $A_2, \ldots, A_k$ centered around $v_1$, each $A_i$ of width
  $2r_i$. Then, for $i=2, \ldots, k$, we either shorten or extend the
  edge $v_{i-1}v_i$ until $D_i$ is contained in its annulus $A_i$; see
  Figure~\ref{fig:strips2annuli}. At each step $i$ we treat the
  remaining path $(v_i, \ldots, v_k)$ and its disks $D_i, \ldots, D_k$
  as a rigid structure that is translated as a whole and in parallel to the heavy edge $v_{i-1}v_i$; see the 
  translation vectors indicated in Figure~\ref{fig:strips2annuli}. In
  the end, each disk $D_i$ is contained in its own annulus $A_i$ and
  thus all disks are still pairwise disjoint. Since we only stretch or
  shrink edges of an $x$-monotone path but do not change any edge
  directions, the whole transformation preserves the previous
  properties of the drawing. Clearly, all disks now lie inside a disk
  $D$ of radius $r=r_1 + 2\sum_{i=2}^k r_i \le 2\sum_{i=1}^k r_i$
  (property~\ref{dsle:item:hp:area}).

	It remains to show the $O(k)$ time bound for drawing $P$.
	Here we store the coordinates of each $v_i$ in $P$ not only relative to the parent node $v_{i-1}$ but also relative to the root $v_1$ of $P$.
	Initially, each disk is placed in its vertical strip as shown in Figure~\ref{fig:v-strips} and the order of the children is selected as either clockwise or counterclockwise as needed. (Recall that changing the direction can be done in constant time.)
	Then for $i=2, \dots, k$ each disk $D_i$ is translated into its annulus $A_i$; see Figure~\ref{fig:strips2annuli}.
	In this process the coordinates of $v_i$ with respect to $v_1$ can become temporarily invalid but the coordinates relative to the predecessor node $v_{i-1}$ remain valid.
	Given the final position of $D_i$ in $A_i$ and the current position of $D_{i+1}$ with respect to $v_i$ we obtain the final position of $D_{i+1}$ in $A_{i+1}$, both with respect to $v_1$ and to $v_i$. 
	The assignment of the coordinates for every node of $P$ thus takes~$O(k)$ time.
\qed
\end{proof}

Combining Lemmas~\ref{lem:heavynode:straight}
and~\ref{lem:heavypath:straight} yields the following theorem:

\begin{theorem}\label{thm:straight}
  Given an unordered tree $T$ with $n$ nodes we can find, in $O(n)$ time and space, a
  crossing-free straight-line drawing of $T$ with perfect angular
  resolution that fits inside a disk $D$ of radius $2\cdot 8^{h(T)}
  n$, where $h(T)$ is the height of the heavy-path decomposition of
  $T$. Since $h(T) \le \log_2{n}$ the radius of $D$ is no more than $2n^4$. 
\end{theorem}

\begin{proof}
  From Lemma~\ref{lem:heavynode:straight} we know that, for each node
  $v$ of a heavy path $P$ at level $j$, the radius of the disk $D$
  containing $v$ and all its light subtrees is $r_v = 8^{h(T)-j}
  l(v)$.
  Lemma~\ref{lem:heavypath:straight} yields that $P = (v_1, \dots, v_k)$ and all its
  descendants can be drawn in a disk of radius $r = 2 \sum_{i=1}^k
  r_{v_i} = 2 \cdot 8^{h(T)-j} \sum_{i=1}^k l(v_i) = 2 \cdot
  8^{h(T)-j} n(P)$, where $n(P)$ is the number of nodes of $P$ and its
  descendants. This holds, in particular, for the heavy path $\hat{P}$
  at the root of $H(T)$.

	It remains to show the linear time and space bound. 
	As indicated in Section~\ref{sec:hpd} the heavy path decomposition is computed in linear time and has linear size. 
	Since the drawing subroutines for nodes and heavy paths in Lemmas~\ref{lem:heavynode:straight} and~\ref{lem:heavypath:straight} both require linear time and are called only once for each node and heavy path, respectively, these steps take $O(n)$ time in total. 
	In the final step we set the coordinates of the root of $T$ to $(0,0)$ and propagate the absolute positions of all nodes from top to bottom. 
	Thus the entire process takes~$O(n)$ time. As we only store a constant amount of information with each node of $T$, it follows that the space needed is also~$O(n)$.
	\qed
\end{proof}

\begin{corollary}\label{cor:area-straight}
	The drawing of $T$ according to Theorem~\ref{thm:straight} requires polynomial area.
\end{corollary}

\begin{proof}
Our first definition of area (the ratio of the area of the smallest enclosing disk over the square of the length of the shortest edge) yields an area value of at most $4 \pi n^8$ for the drawing of $T$ since the shortest edges have length at least 1 and $D$ has radius at most $2n^4$.
In the alternative notions of area defined by the (squared) ratio of the farthest distance of any two nodes (or edges) to the smallest distance of any two nodes (or non-adjacent edges) a similar polynomial area bound holds.
Clearly the farthest distance in both cases is at most the diameter $4n^4$ of $D$.
Furthermore, every child node in the drawing is contained in its own overlap-free disk of radius~1 and hence the closest pair of nodes has distance at least~1.
For the closest pair of edges there is also a lower distance bound of~1. 
In every step of the recursive drawing procedure a subtree $T_u$ is drawn inside a disk $D_u$ with the property that there is an empty outer annulus of width at least~1 in $D_u$.
When composing different subdrawings, this ensures that their edges are kept far enough apart.
Thus it is easy to see by induction that no pair of edges can get closer than distance~1.
\qed
\end{proof}

\section{Straight-line drawings for ordered trees}
\label{sec:exponentialArea}

In many cases, the ordering of the children around each node of a
tree is given; that is, the tree is ordered (or has a fixed
combinatorial embedding).  In the previous section we relied on the
freedom to order subtrees as needed to achieve a polynomial area
bound. Hence that algorithm cannot be applied to ordered trees with
fixed embeddings. As we now show, there are ordered trees that have no
straight-line crossing-free drawings with perfect
angular resolution  \emph{and} polynomial area.

Specifically, we present a class of ordered trees for which any
straight-line crossing-free drawing with perfect angular
resolution requires exponential area. 
We define the \emph{3-legged Fibonacci caterpillar} of length $k$ to be an ordered caterpillar tree $T_k$, whose spine (the subgraph obtained after removing all leaves) is a $k$-node path $P=(p_1, \dots, p_k)$ in which every node $p_i$ has degree 5 in $T_k$, hence three legs.
The embedding of $T_k$ specifies that in every node $p_i$ ($i=2, \dots, k-1$) the edge $p_i p_{i+1}$ is the immediate counterclockwise successor of $p_i p_{i-1}$.
Hence in any straight-line drawing of $T_k$ with perfect angular resolution, the spine is represented as a simple polyline with $k-2$ right turns of $108^\circ$, forming a $72^\circ$ angle between adjacent edges; see Figure~\ref{fig:fibWoolyWorm}.

\begin{figure}[hbt]
\centering
\begin{minipage}[b]{4cm}
\subfloat[\label{fig:fibo-in}]{\includegraphics[width=4cm,page=2]{fibonacciWorm}}\\
\subfloat[\label{fig:fibo-lombardi}]{\includegraphics[width=4cm,page=3]{fibonacciWorm}}
\end{minipage}
\hfil
\subfloat[\label{fig:fibo-straight}]{\includegraphics[page=1,scale=0.8]{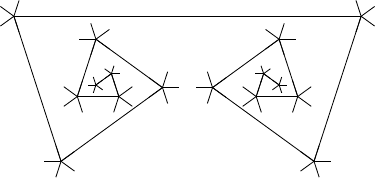}}
\caption{(a) A Fibonacci caterpillar; (b) Lombardi drawing; (c) Straight-line drawing with perfect angular resolution and exponential area.}
\label{fig:fibWoolyWorm}
\end{figure}

We define a \emph{clockwise (counterclockwise) spiral} to be a polyline $(q_1, \dots, q_k)$ such that for any index $3 \le i \le k-1$ the polyline $(q_1, \dots, q_{i})$ lies to the right (left) of the ray $\overrightarrow{q_i q_{i+1}}$.
First, we show that any drawing of the Fibonacci caterpillar contains a large spiral.

\begin{lemma}\label{lem:spiral}
	In any straight-line drawing with perfect angular resolution of $T_k$ the spine $P$ contains a spiral consisting of at least $k/2$ nodes.
\end{lemma}
\begin{proof}
	For $k \le 5$,  because of the required fixed angle turns, either $P=(p_1, \dots, p_k)$ is a clockwise spiral or its reverse $\overline{P}=(p_k, \dots, p_1)$ is a counterclockwise spiral.
	So let $k > 5$. 
	For $i=1, \dots, k-1$ we abbreviate the edge $p_i p_{i+1}$ as $e_i$. 
	We look at sequences $S_i$ of four consecutive edges $(e_i, e_{i+1}, e_{i+2}, e_{i+3})$ of $P$ and distinguish two cases.
	If the extension of edge $e_{i+3}$ into a ray $\overrightarrow{p_{i+3} p_{i+4}}$ intersects $e_i$ or $e_{i+1}$, we say the sequence $S_i$ is \emph{locked}, and otherwise we say it is \emph{open}; see Figure~\ref{fig:spiral-locked-open}.
	Starting from $i=1$ we scan the spine $P$ for the first occurrence $j$ of a locked sequence $S_j$. 
	Then the prefix path $(p_1, \dots, p_{j+3})$ is a clockwise spiral.
		
	\begin{figure}[htbp]
		\centering
		\subfloat[\label{fig:locked1}]{\includegraphics[scale=1,page=2]{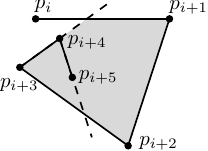}}
		\hfill
		\subfloat[\label{fig:locked2}]{\includegraphics[scale=1,page=3]{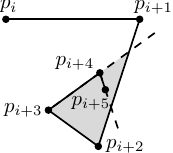}}
		\hfill
		\subfloat[\label{fig:open}]{\includegraphics[scale=1,page=4]{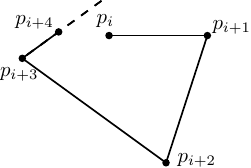}}
		\caption{Two locked edge sequences (a) and (b) and an open sequence (c).}
		\label{fig:spiral-locked-open}
	\end{figure}

	Furthermore, for any $i \ge j$ the sequence $S_i$ is also locked, as can be seen by induction.
	Let $S_i$ be a locked sequence. 
	Then node $p_{i+5}$ lies inside the quadrilateral (or triangle) defined by edges $e_i,e_{i+1}, e_{i+2}$ and the ray $\overrightarrow{p_{i+3} p_{i+4}}$, and due to the angle of $72^\circ$ between $e_{i+3}$ and $e_{i+4}$ the ray  $\overrightarrow{p_{i+4} p_{i+5}}$ must intersect either $e_{i+1}$ or $e_{i+2}$; see Figures~\ref{fig:locked1} and~\ref{fig:locked2}.
	This means that $S_{i+1}$ is also a locked sequence.
	
	By observing that if a sequence $S_i = (e_i,e_{i+1}, e_{i+2}, e_{i+3})$ is locked, then the reverse sequence $\overline{S_i} = (e_{i+3}, e_{i+2}, e_{i+1}, e_i)$ is open, the same reasoning as before yields that the suffix path $(p_j, p_{j+1},\dots, p_k)$ in reverse order $(p_k, \dots, p_{j+1}, p_j)$ is a counterclockwise spiral. Clearly, one of the two spirals contains at least $k/2$ nodes. \qed
\end{proof}

Now that we know that there is a large spiral in $T_k$ we show that drawing the spiral requires exponential area.

\begin{lemma}\label{lem:spiral-area}
	The drawing of a spiral of length $n$ requires exponential area $\Omega(c^n)$ for some $c>1$.
\end{lemma}
\begin{proof}
	Without loss of generality we consider a path $P$ of length $n\ge 6$ that forms a clockwise spiral. 
	Figure~\ref{fig:spiral-area} shows the construction of a minimum-area drawing of $P$. 
	Let the minimum length of any edge be $1$. 
	We draw $e_1$ and $e_2$ with an angle of $72^\circ$ and length $1$ each.
	Every subsequent edge $e_i$ for $3 \le i \le n-1$ is drawn just as long as necessary so that the sequence $S_i$ is open.
	Obviously, no edge can be shortened and increasing any edge only increases the area of the spiral.
	
	\begin{figure}[htbp]
		\centering
			\includegraphics[scale=1]{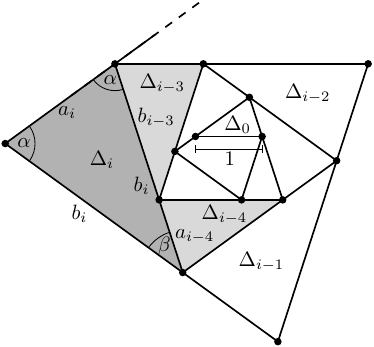}
		\caption{Construction of a minimum-area spiral based on the angle $\alpha=72^\circ$.}
		\label{fig:spiral-area}
	\end{figure}
	
	This procedure creates a sequence of isosceles triangles $\Delta_0, \dots, \Delta_{n-6}$ as indicated in Figure~\ref{fig:spiral-area}.
	Each $\Delta_i$ has two long sides of length $b_i$ and a short side of length $a_i$. 
	The angles opposite the two long sides are $\alpha = 72^\circ$ and the angle opposite the short side is $\beta=36^\circ$.
	By construction of the triangle sequence we obtain the recurrence $b_i = b_{i-3} + a_{i-4}$, which is similar to the definition of the Fibonacci numbers.
	From trigonometry we know that $a_i = \sin 36^\circ / \sin 72^\circ \cdot b_1 = 1/(2 \cos 36^\circ) \cdot b_i \approx 0.618 \cdot b_i$ and that the area of $\Delta_i$ is $A_i = 1/2 \cdot b_i^2 \sin 36^\circ$.
	Using $b_{i} = b_{i-3} + a_{i-4} \ge 2 a_{i-4}$ and $a_i \ge a_{i-4}/\cos 36^\circ$ we can now bound $A_i$ as follows:
	\begin{equation}
		\begin{array}{rcl}
			A_i & = & \frac{1}{2} b_i^2 \sin 36^\circ\\
			    & \ge & 2 \sin 36^\circ a_{i-4}^2\\
			    & \ge & 2 \sin 36^\circ \frac{1}{\cos 36^\circ}^{\lfloor i/4 \rfloor} a_0\\
			    & \ge & 1.236^{\lfloor i/4 \rfloor} a_0.
		\end{array}
	\end{equation}
	Clearly, the smallest disk containing the spiral has area at least $A_{n-6}$ and so by our definition of the area of a drawing the whole spiral has area $\Omega(c^n)$ for $c = \sqrt[4]{1.236} \approx 1.054$. \qed
\end{proof}

By combining Lemmas~\ref{lem:spiral} and~\ref{lem:spiral-area} we immediately obtain the following theorem since drawing the whole Fibonacci caterpillar $T_k$ requires at least as much area as drawing only its spine.

\begin{theorem}\label{thm:straight-expo}
	Any straight-line drawing of the Fibonacci caterpillar $T_k$ with perfect angular resolution requires area $\Omega(c^k)$ for some $c>1$.
\end{theorem}

Similar reasoning was used by Frati~\cite{f-mapuddtofdag-08} to show an exponential lower bound on the area of upward straight-line drawings for ordered trees.
The Fibonacci caterpillar shows that we cannot maintain all
constraints (straight-line edges, crossing-free, perfect angular
resolution, polynomial area) for ordered trees. However, as we show
next, using circular arcs instead of straight-line edges allows us
to respect the other three constraints; see Figure~\ref{fig:fibo-lombardi}.

\section{Lombardi drawings for ordered trees}
\label{sec:lombardi}

In this section, let $T$ be an ordered tree with $n$ nodes. As we have
seen in Section~\ref{sec:exponentialArea}, we cannot find polynomial
area drawings for all ordered trees using straight-line edges. 
However, by using circular arc edges instead of straight-line segments we can achieve all remaining constraints as in the unordered case.
That is, we can find crossing-free circular
arc drawings with perfect angular resolution and polynomial area. Recall that a drawing with circular arcs and perfect angular resolution is called a
Lombardi drawing~\cite{degkn-ldg-12}.

The flavor of the algorithm for Lombardi tree drawings is similar to
our straight-line tree drawing algorithm of
Section~\ref{sec:straight-line-tree}: We first compute a heavy-path
decomposition $H(T)$ for $T$, and then we recursively draw all heavy paths
within disks of polynomial area in a bottom-up fashion. 
More precisely, we ensure the following invariant for the drawing of any heavy path and all its descendants.
\begin{invariant}\label{inv:diskradius}
	A heavy path $P$ at level $j$ of $H(T)$ and all its descendants are drawn inside a disk $D$ of radius $2\cdot
	4^{h(T)-j} n(P)$, where $n(P) = |T_{v}|$ for the root $v$ of $P$.
\end{invariant}
Given the logarithmic height of the heavy path decomposition, this yields a drawing of $T$ with polynomial area.

In Section~\ref{sec:lom-heavy-path}, we describe how to draw a heavy path~$P$ (but not yet its light subtrees) under the assumption that each node of~$P$ is centered in a disk of given radius. 
Subsequently, Section~\ref{sec:light-children} shows how the light subtrees of a heavy-path node~$v$, which are themselves heavy paths of the level below and thus recursively drawn within disks of fixed size according to Invariant~\ref{inv:diskradius}, are placed within the space reserved around~$v$ in the previous step.
These two steps define the drawing of a heavy path $P$ and all its descendants, which we show satisfies Invariant~\ref{inv:diskradius}, and which is then used as a component for the drawing of the parent of~$P$ in~$H(T)$.

\subsection{Drawing heavy paths}
\label{sec:lom-heavy-path}

Let $P=(v_1, \ldots, v_k)$ be a heavy path at level $j$ of the
heavy-path decomposition. Since we will draw $P$ incrementally starting from the leaf and ending with the root of $P$, we assume that the last node $v_k$ is the root of $P$. We denote each edge $v_iv_{i+1}$ by $e_i$.  
Recall that the angle at an intersection point of two circular arcs is
measured as the angle between the tangents to the arcs at that
point. We define the angle $\alpha(v_i)$ for $2 \le i \le k-1$ to be
the angle between $e_{i-1}$ and $e_i$ at node $v_i$ (measured
counter-clockwise).  The angle $\alpha(v_k)$ is defined as the angle
at $v_k$ between $e_{k-1}$ and the light edge $e=v_ku$ connecting
the root $v_k$ of $P$ to its parent $u$. 
Due to the perfect angular resolution requirement for each node $v_i$,
the angle $\alpha(v_i)$ is obtained directly from the number of edges
between $e_{i-1}$ and $e_i$ and the degree $d(v_i)$.

\begin{lemma}\label{lem:heavypath:lombardi}
  Given a heavy path $P=(v_1, \ldots, v_k)$ and a disk $D_i$ of radius
  $r_i$ for the drawing of each $v_i$ and its light subtrees, we can
  draw $P$ with each $v_i$ in the center of its disk $D_i$ inside a
  large disk $D$ in $O(k)$ time 
such that the following properties hold:
  \begin{enumerate}
  \item\label{item:lom:edges} each heavy edge $e_i$ is a
    circular arc that does not intersect any disk other than
    $D_i$ and $D_{i+1}$;
  \item\label{item:lom:light} there is a stub edge incident to $v_k$
    that is reserved for the light edge connecting $v_k$ and its
    parent $u$;
  \item\label{item:lom:disjoint} any two disks $D_i$ and $D_j$ for
    $i\ne j$ are disjoint;
  \item\label{item:lom:angres} the angle between any two consecutive
    heavy edges $e_{i-1}$ and $e_i$ is $\alpha(v_i)$;
  \item\label{item:lom:area} the radius of $D$ is $r=2\sum_{i=1}^k r_i$.
  \end{enumerate}
\end{lemma}

\begin{proof}
  We draw $P$ incrementally starting from the leaf $v_1$ by placing
  $D_1$ in the center $M$ of the disk $D$ of radius $r=2\sum_{i=1}^k
  r_i$. We may assume that $D_1$ is rotated such that the edge $e_1$
  is tangent to a horizontal line at $v_1$ and that it leaves $v_1$ to
  the right. All disks $D_2, \ldots, D_k$ will be placed with their
  centers $v_2, \ldots, v_k$ on concentric circles $C_2, \ldots, C_k$
  around $M$ as shown in Figure~\ref{fig:heavy-lombardi}. 
	The radius of $C_i$ is $r_1 + 2\sum_{j=2}^{i-1} r_j +
  r_i$ so that $D_{i-1}$ and $D_i$ are placed in disjoint annuli separated by the circle $\hat{C}_{i-1}$ of radius $r_1 + 2\sum_{j=2}^{i-1} r_j$. Hence by construction no two disks intersect
  (property~\ref{item:lom:disjoint}). Each disk $D_i$ will be rotated
  around its center such that the tangent to $C_i$ at $v_i$ is the
  bisector of the angle $\alpha(v_i)$.

	\begin{figure}[tb]
    \centering
    \includegraphics{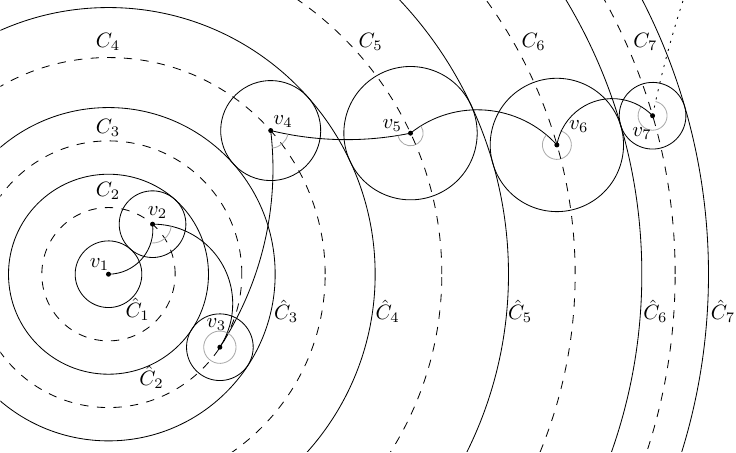}
    \caption{Drawing a heavy path $P$ on concentric circles with
      circular-arc edges. The angles $\alpha(v_i)$ are marked in gray;
      the edge stub to connect $v_7$ to its parent is dotted.}
    \label{fig:heavy-lombardi}
  \end{figure}

  We now describe one step in the iterative drawing procedure that draws
  edge $e_i$ and disk $D_{i+1}$ given a drawing of $D_1, \ldots,
  D_i$. Disk $D_i$ is placed such that $C_i$ bisects the angle
  $\alpha(v_i)$ and hence we immediately obtain the slope of the tangent to $e_i$ at $v_i$.
  This defines a family $\mathcal{F}_i$ of circular arcs emitted from $v_i$ with the same given tangent slope at $v_i$
  that intersect the circle $C_{i+1}$; see
  Figure~\ref{fig:transition-angle}. We consider all arcs from $v_i$
  until their first intersection point with $C_{i+1}$. Observe that
  the intersection angles of $\mathcal{F}_i$ and $C_{i+1}$ bijectively
  cover the full interval $[0, \pi]$, i.e., for any angle $\alpha \in
  [0, \pi]$ there is a unique arc in $\mathcal{F}_i$ that has
  intersection angle $\alpha$ with $C_{i+1}$. Hence we choose for
  $e_i$ the unique circular arc that realizes the angle
  $\alpha(v_{i+1})/2$ and place the center $v_{i+1}$ of $D_{i+1}$ at
  the endpoint of $e_i$. 
	Since the centers of all arcs $a$ in $\mathcal{F}_i$ lie on a line $\ell_i$, we parameterize them as $a=a(t)$ by a parameter $t \in \mathbb{R} \cup \{\infty\}$ that yields the corresponding circle center on $\ell_i$. 
	Then we consider the angle of the tangents to $a(t)$ and the circle $C_{i+1}$ in their first intersection point $p(t)$ and set it equal to $\alpha(v_{i+1})/2$. 
	Solving this equation for $t$ requires finding the roots of a polynomial of bounded degree, which we assume to be possible in constant time.
	We store the resulting arc center and radius together with $e_i$.
  We continue this process until the last disk
  $D_k$ is placed. This drawing of $P$ realizes the angle
  $\alpha(v_i)$ between any two heavy edges $e_{i-1}$ and $e_i$
  (property~\ref{item:lom:angres}). For the edge from $v_k$
  to its parent $u$ we can only reserve a stub whose tangent at $v_k$ has
  a fixed slope
  (property~\ref{item:lom:light}).
	The only information that we have about the edge $v_k u$ is that it belongs to the family $\mathcal{F}_k$ of arcs that intersect the circle $\hat{C}_k$ and have the given tangent at $v_k$. 
	This ambiguity does not cause problems in the subsequent steps though, and hence we can reserve all of the possible arcs simultaneously.
	Figure~\ref{fig:heavy-lombardi} shows an example.

  \begin{figure}[tbp]
  	\centering
  		\includegraphics[width=.6\textwidth]{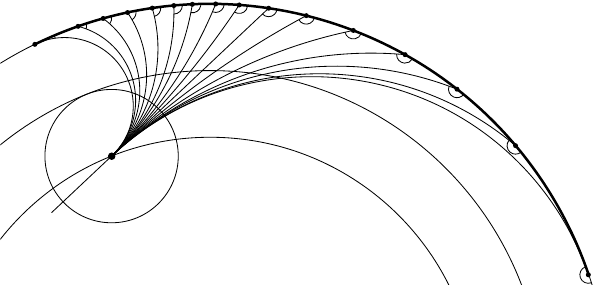}
  	\caption{Any angle $\alpha \in [0,\pi]$ can be realized.}
  	\label{fig:transition-angle}
  \end{figure}

  Each edge $e_i$ is contained in the annulus between $C_i$
  and $C_{i+1}$ and thus does not intersect any other edge of the
  heavy path or any disk other than $D_i$ and $D_{i+1}$
  (property~\ref{item:lom:edges}). Furthermore, the disk $D$ with
  radius $r=2\sum_{i=1}^k r_i$ indeed contains all the disks $D_1,
  \ldots, D_k$ (property~\ref{item:lom:area}). 

	It remains to show the time bound for computing the drawing of $P$.
	Similarly to drawing heavy paths in Section~\ref{sec:straightLine}, we store the position of each node $v_i$ in polar coordinates relative to its predecessor and relative to the center $M$ of $D$. This avoids the need to update the positions of all descendants in every step and allows to assign the final absolute coordinates in a top-down traversal of $T$.
	Given the position of node $v_{i}$ (with respect to $M$) we can compute the position of $v_{i+1}$ with respect to $M$ in constant time as described above. 
	Once all nodes of $P$ are placed, we additionally set the coordinates of each node $v_i$ with respect to its parent $v_{i+1}$. 
	The required time is $O(k)$.	
	\qed
\end{proof}

Lemma~\ref{lem:heavypath:lombardi} showed how to draw a heavy path $P$
with prescribed angles between the heavy edges and an edge stub to
connect it to its parent. Since the root $v$ of each heavy path $P$ (except the path
at the root of $H(T)$) is the light child of a node on the previous
level of $H(T)$, the light edge
from $v$ to its parent is actually drawn when placing the
light subtrees of a node, the topic of the next subsection.

\subsection{Drawing light subtrees}
\label{sec:light-children}

Once the heavy path $P$ is drawn as described above, it remains to
place the light subtrees of each node $v_i$ of $P$. For each node
$v_i$ the two heavy edges incident to it partition the disk $D_i$ into
two regions. We call the region that contains the larger conjugate
angle the \highlight{large zone} of $v_i$ and the region that contains the
smaller conjugate angle the \highlight{small zone}. 
If both angles are equal to $\pi$, then we can consider both regions as small zones.
For the root node $v_k$ of $P$ we only know one heavy edge to $v_{k-1}$, whereas the light edge to its parent $u$ is not yet fully determined. 
In this case we define the two zones of $v_k$ as the regions between the heavy edge and the leftmost/rightmost possible arc for the light edge $v_k u$.
Since the light subtrees of $v_k$ will be placed in the small and large zones of $v_k$, the node $v_k$ can always be connected to its parent $u$ by an arc that does not intersect any edge in $T_{v_k}$.
We say that $v_k$ is \emph{exposed} to its parent.
Our approach in this section proceeds in two steps. 
First, we find a disjoint placement of the child disks in the small and large zone. 
In the second step, we actually draw the light edges from $v_i$ to all its light children.

For a node $v_i$ at level $j$ of $H(T)$ we define the radius $r_i$ of
$D_i$ as $r_i = 4^{h(T)-j} (1+\sum_{u \in L(v_i)} |T_u|) = 4^{h(T)-j}
l(v_i)$. All light children of $v_i$ are at level $j+1$ of $H(T)$ and
thus by inductively assuming that Invariant~\ref{inv:diskradius} holds, every light child $u$ of $v_i$ and its subtree is drawn in a
disk of radius $r_u = 2\cdot 4^{h(T)-j-1} |T_u|$. Thus we know that
$r_u \le r_i/2$; in fact, we even have $\sum_{u \in L(v_i)} r_u \le
r_i/2$.

\paragraph{Light subtrees in the small zone.}

Depending on the angle $\alpha(v_i)$, the small zone of a disk $D_i$
might actually be too narrow to directly place the light subtrees in
it. 
Therefore, we define the \emph{extended small zone} as the area bounded by $e_{i-1}$, $e_i$, $\hat{C}_{i-1}$, $\hat{C}_{i}$, and the horizontal ray to $-\infty$ through $v_1$.
Fortunately, we can always place another disk $D'$ of radius at most $r_i/2$ in this extended small zone such that $D'$ touches $e_{i-1}$ and $e_i$
and does not intersect any other previously placed disk; see
Figure~\ref{fig:small-zone}. For a given radius of $D'$ the position of the center of $D'$ with respect to $v_i$ can be computed in constant time. If there is a single child $u$ in the
small zone then $D' = D_u$ and we are done. The next lemma shows how
to place more than one child. 
Let $l \ge 2$ be the number of light children of $v_i$ to be placed in the (extended) small zone.
We say that the disks $D'_1, \dots, D'_l$ are \emph{correctly placed} in the (extended) small zone if their interiors are mutually disjoint and if every point inside any disk $D'_i$ can be reached by a circular arc from $v_i$ with given slope at $v_i$ such that the arc does not intersect any other disk $D'_j$ for $j\ne i$.

\begin{figure}[bt]
  \centering
  \includegraphics[page=2]{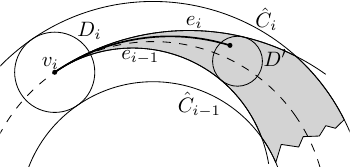}
  \caption{Placing a single disk $D'$ in the extended small zone of
    $D_i$ (shaded gray).}
  \label{fig:small-zone}
\end{figure}

\begin{lemma}\label{lem:one-size-fits-all}
  If a single disk $D'$ of radius $r'$ can be placed in the (possibly
  extended) small zone of the disk $D_i$, then we can correctly place
  any sequence of $l$ disks $D'_1, \ldots, D'_l$ with radii $r_{1}',
  \ldots, r_{l}'$ and $\sum_{i=1}^l r_{i}' = r'$ in the (extended)
  small zone of $D_i$. This can be done in $O(l)$ time.
\end{lemma}

\begin{proof}
  The idea of the algorithm for placing the $l$ disks is to first
  place the disk $D'$ in the small zone as before. The disks $D'_1,
  \ldots, D'_l$ will then be placed within $D'$ so that no additional
  space is required.

  In the first step of the recursive placement algorithm we either
  place $D'_1$ or $D'_l$ (whichever has smaller radius) and a disk
  $D''$ containing the remaining sequence of disks $D'_2, \ldots,
  D'_l$ or $D'_1, \ldots, D'_{l-1}$, respectively. Without loss of
  generality, let $r_{1}' \le r_{l}'$ and thus in particular $r_1' \le
  r'/2$. In order to fit inside $D'$ the disks $D'_1$ and~$D''$ must
  be placed with their centers on a diameter of $D'$; see
  Figure~\ref{fig:many-disks}a. The degree of freedom that we have is
  the rotation of that diameter around the center of $D'$. Then the
  locus of the tangent point of $D'_1$ and $D''$ is a circle $\hat{C}$
  of radius $r' - 2r_1'$ around the center of $D'$; see
  Figure~\ref{fig:many-disks}b. For any given tangent slope at~$v_i$, in particular the slope required for the edge from $v_i$ to the light child in $D'_1$, there are exactly two circular arcs
  $a_1$ and $a_2$ that are tangent to $\hat{C}$. 
	They can be computed in constant time. 
	Let the two points of tangency on $\hat{C}$ be $p_1$ and $p_2$. 
	Now we rotate $D'_1$ and $D''$ such that their point of tangency coincides with either $p_1$ or $p_2$ depending on which of them yields the correct  embedding order of $D'_1$ and $D''$ around $v_i$. 
	Clearly, $a_1$ or $a_2$ are also tangent to $D'_1$ and $D''$ now. 
	Assume we choose $p_1$ and the corresponding arc $a_1$ as in Figure~\ref{fig:many-disks}b. 
  We claim that we can connect any point in $D'_1$ to $v_i$ with the unique circular arc of the required slope at node $v_i$ without creating any edge crossings. 
	(We will describe the exact placement of that arc later.)
  As in the proof of Lemma~\ref{lem:heavypath:lombardi}, there is a family $\cal F$ of circular arcs that pass through $v_i$ with the given slope.
	We consider the subset $\mathcal{F}' \subset \mathcal{F}$ that intersects disk $D'_1$ and thus can be used as basis for the edge from $v_i$ to the light child in $D'_1$. 
  Any such arc stays inside the horn-shaped region $\Upsilon$ that encloses $D'_1$ and is formed by a boundary arc $b$ of the small zone and $a_1$ before it reaches $D'_1$. 
  Assume to the contrary that there is an arc $a \in \mathcal{F}'$ that does not completely lie inside $\Upsilon$ before reaching $D'_1$. 
	The arc $a$ cannot intersect $a_1$ in a point other than $v_1$ since both $a$ and $a_1$ belong to~$\cal F'$.
	So $a$ must intersect the other boundary arc $b$ of $\Upsilon$. 
	However, since $a$ intersects $b$ in $v_i$ and lies inside $\Upsilon$ in some $\varepsilon$-neighborhood of $v_i$ it would have to intersect $b$ at least three times in order to reach a point of $D'_1 \subset \Upsilon$.
	This is a contradiction. 
	Since $a_1$ separates $D'_1$ from~$D''$, none of the arcs in $\cal F'$ nor $D'_1$ can interfere with any of the disks $D'_2, \ldots,
  D'_l$ and their respective edges as long as those disks stay inside $D''$
  or the edges connect to points in $D''$.
  \begin{figure}[tb]
    \centering
    \includegraphics[page=2]{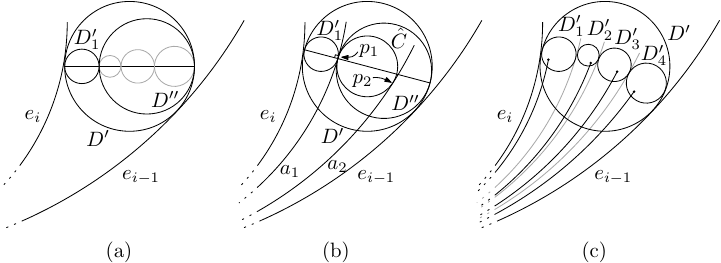}
    \caption{Placing disks $D'_1$ and $D''$ inside the disk $D'$.}
    \label{fig:many-disks}
  \end{figure}

  For placing $D'_2, \ldots, D'_l$ we recursively apply the same
  procedure again, now using~$D''$ as the disk $D'$ and $a_1$ as one
  of the boundary arcs. Then after $l$ steps, we have disjointly
  placed all disks $D'_1, \ldots, D'_l$ inside the disk $D'$ such that
  their order respects the given tree order and no two edges can possibly
  intersect. In other words they are correctly placed and each step can be performed in constant time.
	Figure~\ref{fig:many-disks}c gives an example. \qed
\end{proof}

We required that the edges $e_{i-1}$ and $e_i$ are tangent
to $D'$, which is possible only for an opening angle $\alpha$ of the
small zone of at most $\pi$. For any angle~$\alpha \le \pi$ the arcs
$a_1$ and $a_2$ always stay within the extended small zone and form
at most a semi-circle. This does not hold for $\alpha > \pi$.

\paragraph{Light subtrees in the large zone.}

Placing the light subtrees of a node $v_i$ in the large zone of
$D_i$ must be done slightly different from the algorithm for the small
zone since Lemma~\ref{lem:one-size-fits-all} holds only for opening
angles of at most $\pi$. On the other hand, the large zone does not
become too narrow and there is no need to extend it beyond~$D_i$. Our
approach splits the large zone incident to the heavy-path node $v_i$ into two parts that again have an opening angle of at most $\pi$ so that we can apply
Lemma~\ref{lem:one-size-fits-all} and draw all the children of $v_i$ accordingly.

Let $l$ be the number of light subtrees in the large zone of $D_i$. We
first place a disk $D'$ of radius at most $r_i/2$ that touches
$v_i$ and whose center lies on the line bisecting the opening
angle of the large zone. The disk $D'$ is large enough to contain the
disjoint disks $D'_1, \ldots, D'_l$ for the light subtrees of $v_i$
along its diameter. We need to distinguish whether $l$ is even or
odd. For even $l$ we create a container disk $D''_1$ for disks $D'_1,
\ldots, D'_{l/2}$ and a container disk $D''_2$ for $D'_{l/2 +1},
\ldots, D'_l$. Now $D''_1$ and~$D''_2$ can be tightly packed on the
diameter of $D'$. Using a similar argument as in
Lemma~\ref{lem:one-size-fits-all} we separate the two disks by a
circular arc through $v_i$ that is tangent to the bisector of
$\alpha(v_i)$ in $v_i$. Since $D'$ is centered on the bisector this is
possible even though the actual opening angle of the large zone is
larger than $\pi$.  If $l$ is odd, we create a container disk $D''_1$
for disks $D'_1, \ldots, D'_{\lfloor l/2 \rfloor}$ and a container
disk $D''_2$ for $D'_{\lceil l/2 \rceil +1}, \ldots, D'_l$. The median
disk $D'_{\lceil l/2 \rceil}$ is not included in any container. Then
we apply Lemma~\ref{lem:one-size-fits-all} to $D'$ and the three disks
$D''_1, D'_{\lceil l/2 \rceil}, D''_2$ along the diameter of $D'$; see
Figure~\ref{fig:children-large-zone}a. The separating circular arcs in
$v_i$ are again tangent to the bisector of $\alpha(v_i)$, which is,
since $l$ is odd, also the correct slope for the circular arc
connecting~$v_i$ to the median disk $D'_{\lceil l/2 \rceil}$.

In both cases we split the large zone and the sequence of light
subtrees to be placed into two parts that each have an opening angle
at $v_i$ of at most $\pi$ between a separating circular arc and the
edge $e_{i-1}$ or $e_i$, respectively. Next, we move $D''_1$ and~$D''_2$ along the separating circular arcs keeping their tangencies
until they also touch the edge $e_{i-1}$ or $e_i$, respectively. Then
we can apply Lemma~\ref{lem:one-size-fits-all} to both container disks
and thus place all light subtrees in the large zone; see
Figure~\ref{fig:children-large-zone}b. The splitting of the large zone involves finding tangent arcs to at most three disks and thus takes constant time. Combining this with the running time in Lemma~\ref{lem:one-size-fits-all} for the two small subinstances all $l$ disks in the large zone can be placed in $O(l)$ time.

\begin{figure}[tb]
  \centering
  \includegraphics{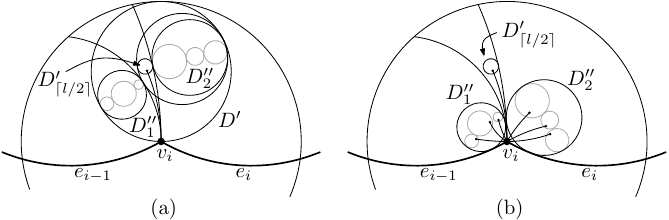}
  \caption{Placing light subtrees in the large zone by first splitting
  it into two parts (a) and then applying the algorithm for small
  zones to each part (b).}
  \label{fig:children-large-zone}
\end{figure}

\paragraph{Drawing light edges}

The final missing step is how to actually connect a heavy node $v_i$
to its light children given a position of $v_i$ and positions with respect to $v_i$ of all disks containing the light subtrees of $v_i$. Let $u$ be a light child of $v_i$
and let $D_u$ be the disk containing the drawing of $T_u$. When
placing the disk $D_u$ in the small or large zone of $v_i$ we made
sure that a circular arc from $v_i$ with the tangent required for
perfect angular resolution at $v_i$ can reach any point inside $D_u$
without intersecting another edge or disk. 

On the other hand, we know by Lemma~\ref{lem:heavypath:lombardi} that
$u$ is placed in the outermost annulus of $D_u$ and that it has reserved a stub for the edge $e=uv_i$. This stub represents all arcs in $u$ that share the tangent for $e$ required to obtain perfect angular resolution in $u$. 
Let $C_u$ be the circle that is the locus of $u$ if we rotate $D_u$ and the drawing of $T_u$ around the center of~$D_u$.

There is again a family $\mathcal{F}$ of circular arcs with the
required tangent at $v_i$ that all lead towards $D_u$ and intersect the
circle $C_u$. As observed in Lemma~\ref{lem:heavypath:lombardi} the
intersection angles formed between $\mathcal{F}$ and $C_u$ bijectively
cover the full interval $[0,\pi]$, i.e., for any angle $\alpha \in
[0,\pi]$ there is a unique circular arc in $\mathcal{F}$ that has an
intersection angle of~$\alpha$ with $C_u$. In order to correctly
attach $u$ to $v_i$ we first compute the arc $a$ in $\mathcal{F}$ that
realizes an intersection angle of $\alpha(u)/2$ with $C_u$, where
$\alpha(u)$ is the angle between $e$ and the heavy edge from $u$ to
its heavy child that is required for perfect angular resolution at
$u$. This arc $a$ can be computed in constant time similarly to computing a heavy-path edge in Lemma~\ref{lem:heavypath:lombardi}.
Let $p$ be the intersection point of $a$ with $C_u$.  Then
we rotate $D_u$ and the drawing of $T_u$ around the center of $D_u$
until $u$ is placed at $p$; see node~$v_7$ in
Figure~\ref{fig:heavy-lombardi}. 
This rotation is actually realized by setting the coordinates of $u$ with respect to its parent $v_i$ to those of $p$. 
We also store with $u$ the rotation angle between the new position of $D_u$ and its neutral position.
Since the stub of $u$ for $e$ also has an angle of $\alpha(u)/2$ with~$C_u$, the arc~$a$ indeed realizes the edge $e$ with the required angles for perfect angular resolution in both $u$ and $v_i$. 
Furthermore, $a$ does not enter the disk bounded by~$C_u$ and hence it does not intersect any part of the drawing of $T_u$ other than $u$. 

We can summarize our results for drawing the light subtrees of a node
as follows:

\begin{lemma}\label{lem:lom:lightchildren}
  Let $v$ be a node of $T$ at level $j$ of $H(T)$ with two incident
  heavy edges. For every light child $u \in L(v)$ assume there is a
  disk $D_u$ of radius $r_u = 2\cdot 4^{h(T)-j-1} |T_u|$ that contains
  a fixed drawing of $T_u$ with perfect angular resolution and such
  that $u$ is exposed to its parent $v$. Then we can
  construct in $O(d(v))$ time a drawing of $v$ and its light subtrees inside a disk $D$,
  potentially with an extended small zone, such that the following
  properties hold:
  \begin{enumerate}
  \item\label{item:edge} the edge between $v$ and any light child $u
    \in L(v)$ is a circular arc that does not intersect any
    disk other than $D_u$;
  \item\label{item:heavyedge} the heavy edges do not intersect any disk $D_u$;
  \item\label{item:disks} any two disks $D_u$ and $D_{u'}$ for $u \ne
    u'$ are disjoint;
  \item\label{item:angres} the angular resolution of $v$ is $2\pi/d(v)$;
  \item\label{item:area} the disk $D$ has radius $r_v = 4^{h(T)-j} l(v)$.
  \end{enumerate}
\end{lemma}

Now we have all ingredients for drawing the entire tree $T$ based on its heavy-path decomposition. We combine Lemmas~\ref{lem:heavypath:lombardi}
and~\ref{lem:lom:lightchildren} to recursively obtain a Lombardi drawing of $T$ in a bottom-up fashion. In the final step, we set the coordinates of the root of $T$ to $(0,0)$ and propagate the absolute node and edge positions downward using the relative positions and rotation angles stored during the recursive calls. We conclude with the following theorem:

\begin{theorem}\label{thm:lombardi}
  Given an ordered tree $T$ with $n$ nodes we can find in $O(n)$ time and space a
  crossing-free Lombardi drawing of $T$ that preserves the embedding
  of $T$ and fits inside a disk $D$ of radius $2\cdot 4^{h(T)} n$, where $h(T)$
  is the height of the heavy-path decomposition of $T$. Since $h(T) \le
  \log_2 n$ the radius of $D$ is no more than $2n^3$.
\end{theorem}

\begin{corollary}\label{cor:area-lombardi}
	The drawing of $T$ according to Theorem~\ref{thm:lombardi} requires polynomial area.
\end{corollary}

\begin{proof}
	Since the shortest edges have again length at least 1, Theorem~\ref{thm:lombardi} implies that the area of the Lombardi drawing of $T$ is at most $4 \pi n^6$ according to our first area measure.
	Exactly the same arguments as used in Corollary~\ref{cor:area-straight} yield again that the polynomial area bounds continue to hold for the two alternative definitions of area based on the (squared) distance ratio of the farthest pair of nodes (or edges) to the closest pair of nodes (or non-adjacent edges), where in this case the farthest pair has distance at most $4n^3$ and the closest pair again at least distance~$1$.
	\qed
\end{proof}

Figure~\ref{fig:exmpl-lombardi} shows a drawing of the ordered tree in Figure~\ref{fig:heavyPath}
according to our method. Instead of asking for perfect
angular resolution, the same algorithm can also be used to construct a
circular-arc drawing of an ordered tree with an arbitrary given assignment of angles
between consecutive edges around each node that add up to $2\pi$. The
drawing remains crossing-free and fits inside a disk of radius~$O(n^3)$.

\section{Implementation Details}
\label{sec:implement}

Since tree drawings with perfect angular
resolution are also of practical importance, we have
implemented a basic version of our straight-line drawing algorithm.
The algorithm, whose area is polynomially bounded, from a practical viewpoint
is still far from desirable.  In particular, as
Figure~\ref{fig:implementSimple} illustrates, there is significant
space left between sibling nodes.
As Figure~\ref{fig:implementBetter}
demonstrates, with some simple heuristical refinements, far
better use of space can be achieved.

  \begin{figure}[h]
    \centering
    \subfloat[\label{fig:implementSimple}Layout drawn by the unmodified straight-line tree drawing algorithm.
		Although polynomially bounded, the area
		is so large that smaller features of the drawing are difficult to	see.]{\makebox[.48\textwidth]{\includegraphics[page=1,width=.4\textwidth]{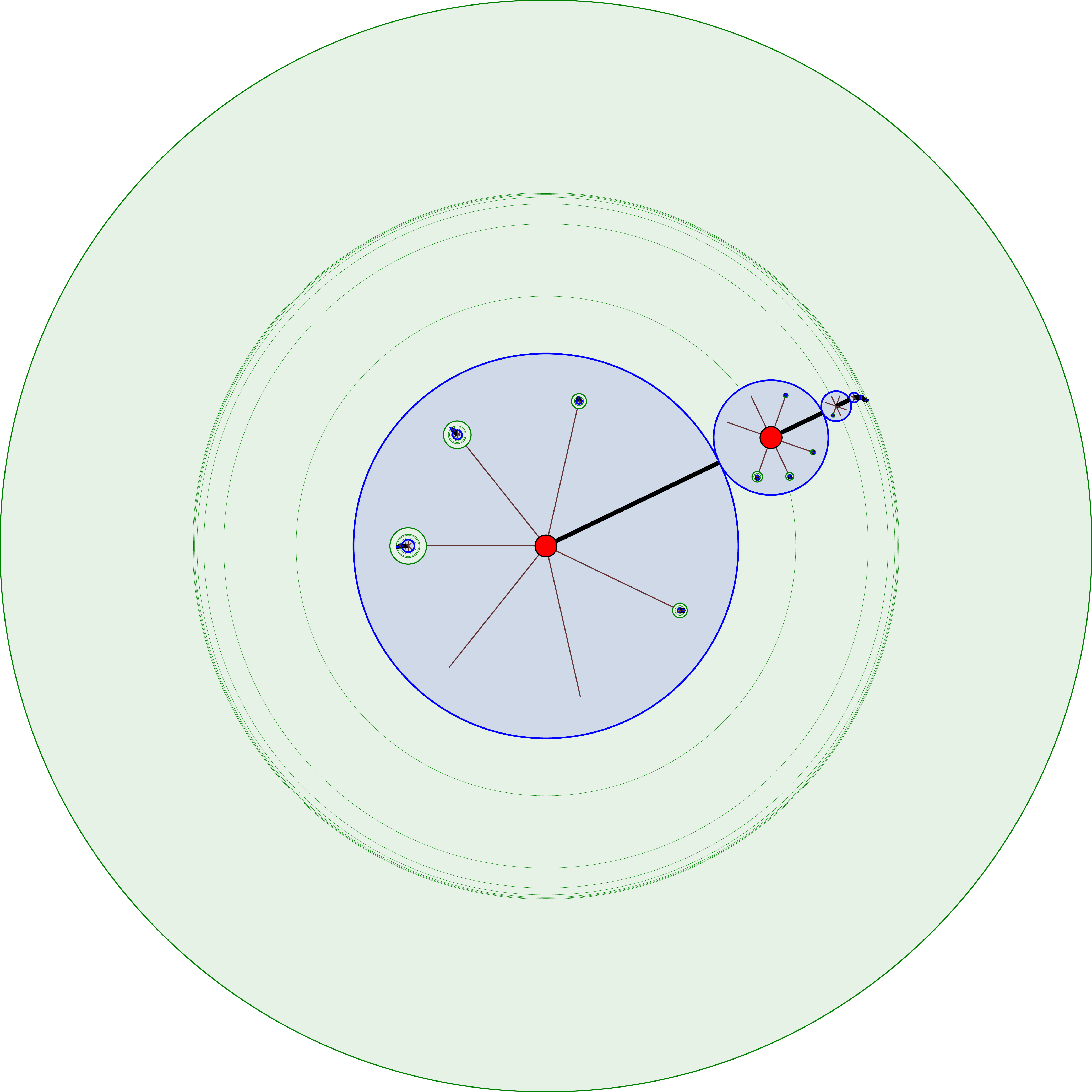}}}
    \hfill
    \subfloat[\label{fig:implementBetter}A space-optimized drawing that still maintains the
stated guarantees.]{\includegraphics[page=1,width=.4\textwidth]{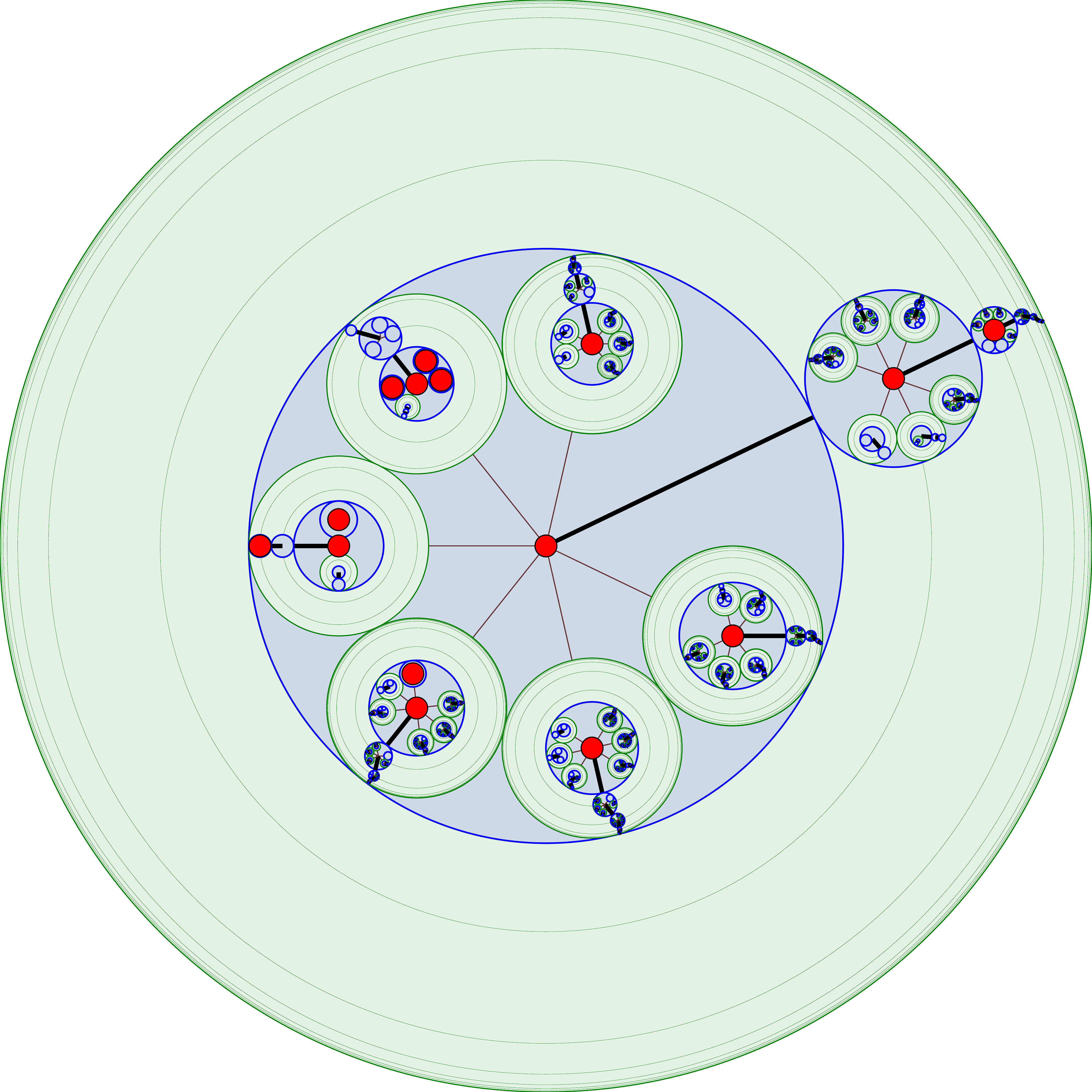}}
    \caption{A partial snapshot of a tree drawing.}
    \label{fig:implement}
  \end{figure}

We highlight a few straightforward space-saving improvements to the algorithm that still ensure the same area bound.
In the original construction, only large nodes are placed on the outer region with the smaller nodes placed inside the inner annulus.  
By continuing with a greedy approach of repeatedly inserting the next largest node in the outer region, skipping the spoke associated with the heavy edge, until no more nodes fit, and filling the remaining spokes with the smaller children, we can insert more nodes into the outer region.  
Moreover, the radii for many of the subtrees are far smaller than necessary.
After laying out the positions of each of the light subtrees, 
we increase their radii so their disk fits maximally within 
their wedge region, thus using considerably more of
the allocated space.
Noting that the heavy path also does not completely fill the disk 
associated with its head node, we also increase this radius as a 
constant factor after having laid out the main drawing.
Figures~\ref{fig:fillOuter} and~\ref{fig:exampleIllustrations}
provide further illustrations of these improvements.

  \begin{figure}[htb]
    \centering
    \subfloat[Partial layout drawn by the unmodified straight-line tree drawing algorithm that
places only large nodes in the outside annulus.]{\makebox[.48\textwidth]{\includegraphics[page=1,width=.4\textwidth]{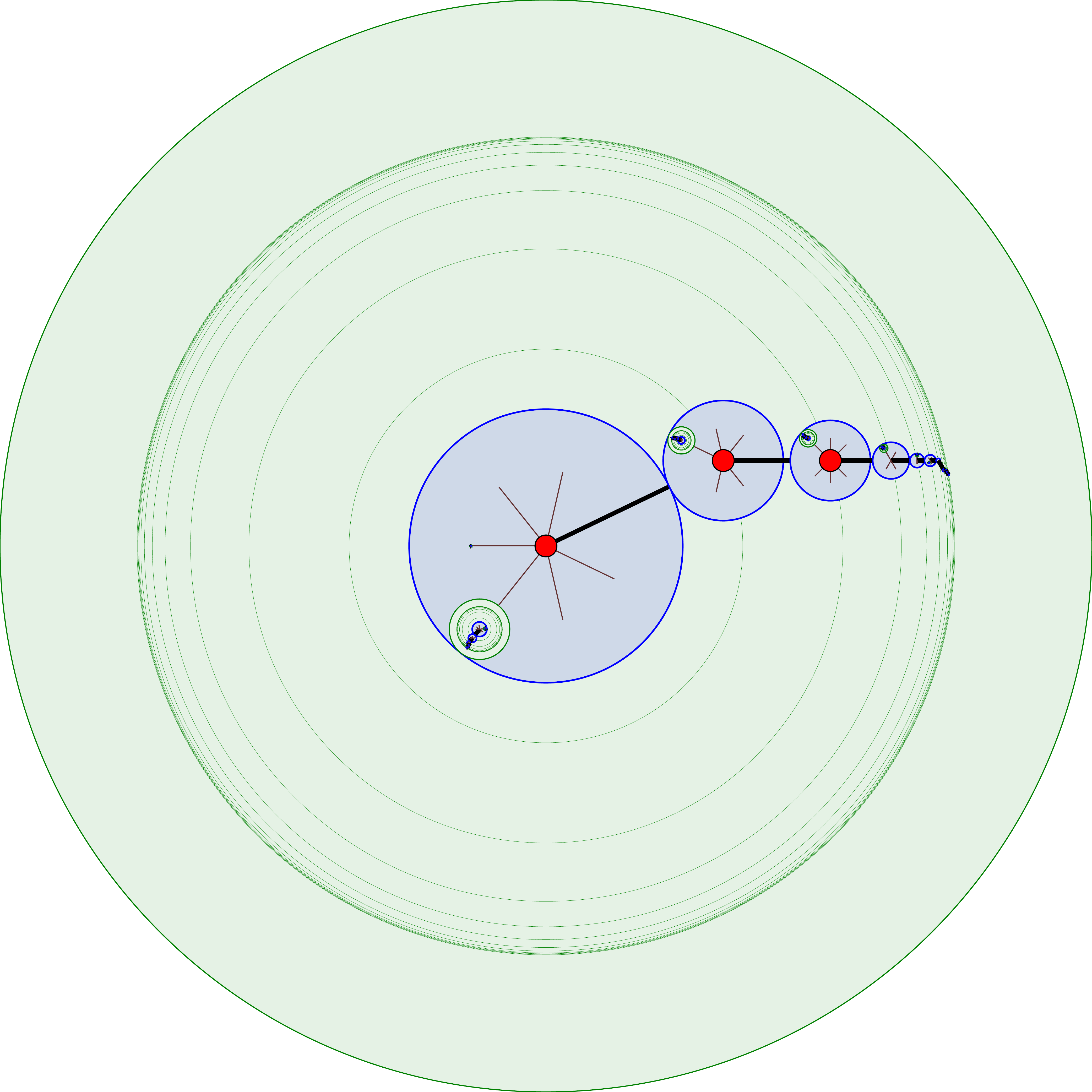}}}
    \hfill
    \subfloat[The same tree but with space-filling optimization in place.]{\includegraphics[page=1,width=.4\textwidth]{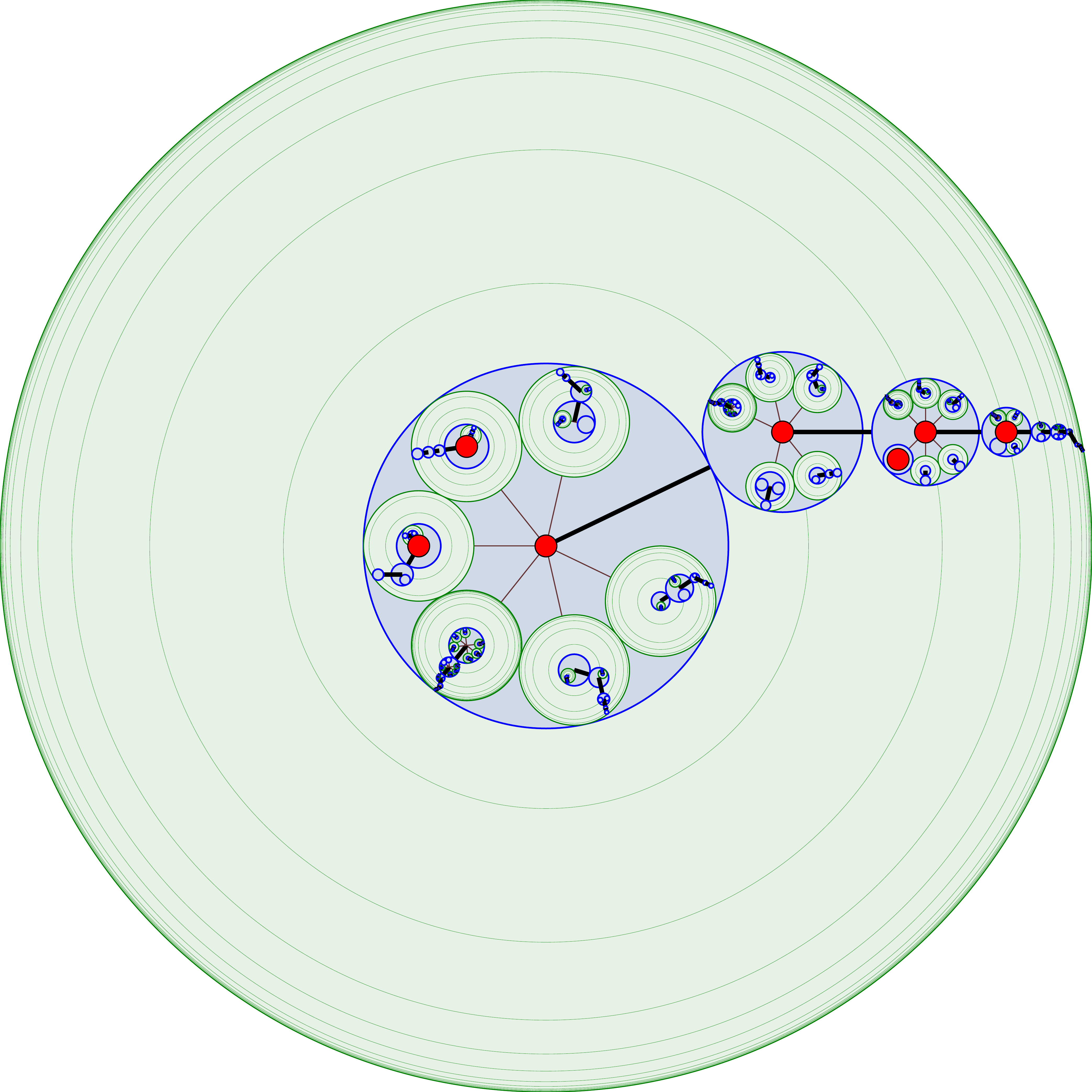}}
    \caption{A partial snapshot of a tree drawing demonstrating filling the disk associated
with the light subtree.}
    \label{fig:fillOuter}
  \end{figure}

  \begin{figure}[htb]
    \centering
    \subfloat[The Fibonacci caterpillar drawn as an unordered tree.]{\makebox[.48\textwidth]{\includegraphics[page=1,width=.4\textwidth]{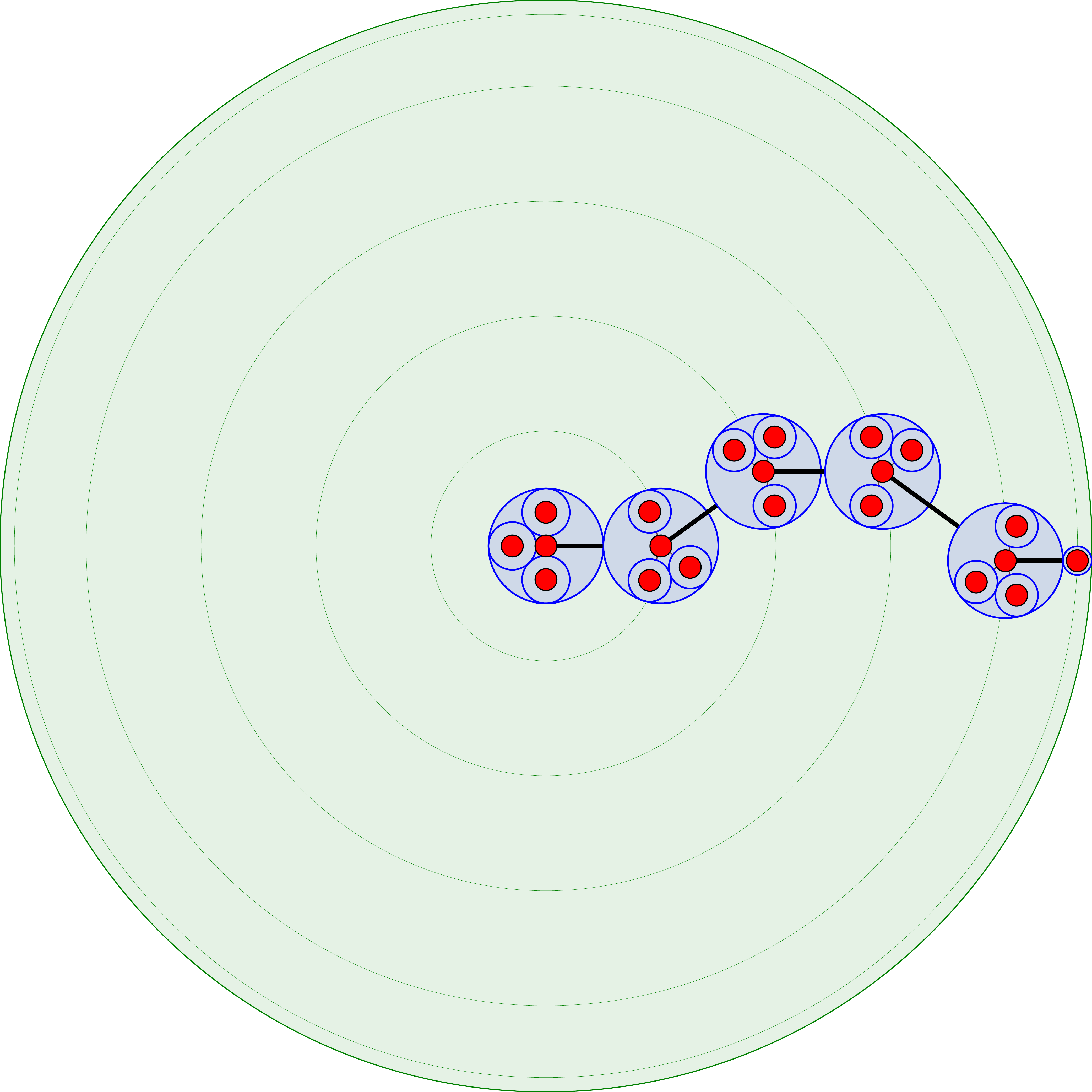}}}
    \hfill
    \subfloat[A 5-ary tree with different weight distributions per child.]{\includegraphics[page=1,width=.4\textwidth]{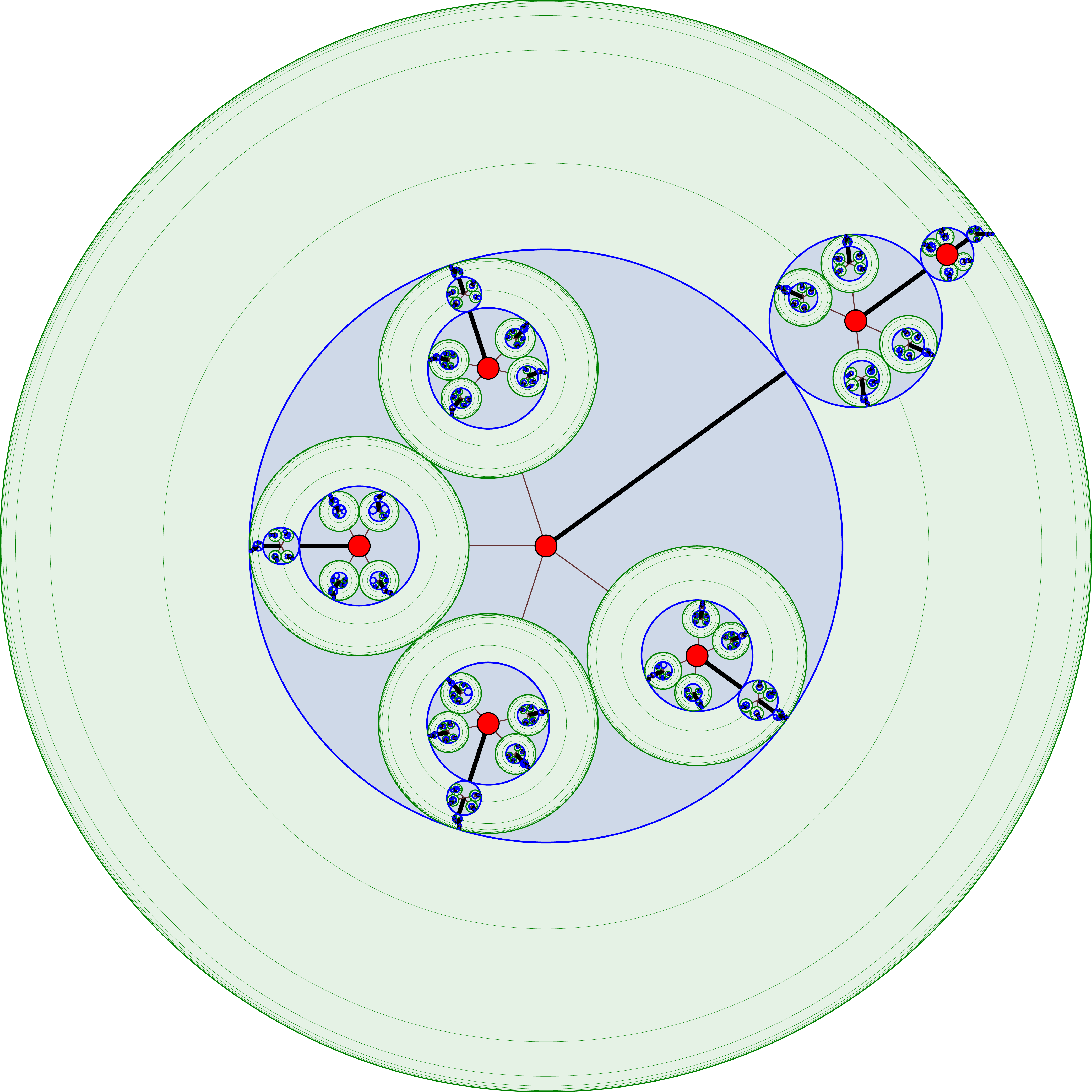}}
    \caption{Example illustrations.}
    \label{fig:exampleIllustrations}
  \end{figure}

\section{Conclusion and Closing Remarks}
\label{sec:conclusion}
We have shown that straight-line drawings of trees 
with perfect angular resolution and polynomial area can be efficiently computed, by carefully
ordering the children of each node and by using a style similar to
balloon drawings in which the children of any node are placed on two
concentric circles rather than on a single circle. However, using our
Fibonacci caterpillar example we also showed that this combination of
straight-line edges, perfect angular resolution, and polynomial area can
no longer be achieved if the order of the children of each node is fixed. Fortunately, for ordered trees with a fixed embedding, Lombardi drawings (in which edges are drawn as circular arcs) allow us to retain the other desirable qualities of absence of crossings, polynomial area, and perfect angular
resolution.

In addition to needing to implement the algorithm for an
ordered tree,
there remain further improvements to the basic implementation 
for the unordered tree discussed in Section~\ref{sec:implement}.
Since our intent was to highlight the key heavy path breakdown in our algorithm,
even when the heavy child could fit as one of the node's light children,
we opted to place the heavy child separately, requiring more
space than generally necessary.

Several problems in the study of Lombardi drawings of trees remain open.
For example, our polynomial area bounds are likely not tight. In fact, recently Halupczok and Schulz~\cite{hs-pbwpaoa-11} showed that any unordered $n$-node tree can be drawn within a disk of radius $n^{3.0367}$ using straight-line edges with perfect angular resolution. 
Moreover, our method
is impractically complex. It would be of interest to find simpler
Lombardi drawing algorithms that achieve perfect angular resolution
for more limited classes of trees, such as binary trees, with better area bounds.
Finally, there are many open problems in the area of plane Lombardi drawings of planar graphs.

\subsection*{Acknowledgments}

This research was supported in part by the National Science
Foundation under grants CCF-0545743, CCF-1115971 and CCF-0830403,
by the
Office of Naval Research under MURI grant N00014-08-1-1015, and by the
German Research Foundation under grant NO 899/1-1.

{\small
\raggedright
\setlength{\itemsep}{0pt} 
\bibliographystyle{abuser}
\bibliography{lombardi}
}

\end{document}